\documentclass[11pt]{article}
\usepackage{style}
\usepackage[left=1in,bottom=1in,top=1in,right=1in]{geometry}
\hypersetup{%
	colorlinks   = true, %
	urlcolor     = blue, %
	linkcolor    = blue, %
	citecolor    = blue,  %
}

\title{Exploring Facets of Language Generation in the Limit}
\author{Moses Charikar\thanks{Stanford University. Email: \texttt{moses@cs.stanford.edu.}}
\and
Chirag Pabbaraju\thanks{Stanford University. Email: \texttt{cpabbara@cs.stanford.edu.}}}
\date{\today}

\begin{document}

\maketitle
\begin{abstract}

    The recent work of %
    \citet{kleinberg2024language} provides a concrete model for language generation in the limit: given a sequence of examples from an unknown target language, the goal is to generate new examples from the target language such that no incorrect examples are generated beyond some point.
    In sharp contrast to strong negative results for the closely related problem of language identification, 
    they establish positive results for language generation in the limit for all countable collections of languages. Follow-up work by %
    \citet{raman2024generation} 
    studies bounds on the number of distinct inputs required by an algorithm before correct language generation is achieved --- namely, whether this is a constant for all languages in the collection (uniform generation) or a language-dependent constant (non-uniform generation).

    We show that every countable language collection has a generator which has the stronger property of non-uniform generation in the limit. However, while the generation algorithm of \citep{kleinberg2024language} can be implemented using membership queries, we show that any algorithm cannot non-uniformly generate 
    even for collections of just two languages,
    using only membership queries.

    We also formalize the tension between validity and breadth in the generation algorithm of \citep{kleinberg2024language} by introducing a definition of \textit{exhaustive} generation, and show a strong negative result for exhaustive generation. Our result shows that a tradeoff between validity and breadth is inherent for generation in the limit. We also provide a precise characterization of the language collections for which exhaustive generation is possible. Finally, inspired by algorithms that can choose to obtain feedback, we consider a model of uniform generation with feedback, completely characterizing language collections for which such uniform generation with feedback is possible in terms of an abstract complexity measure of the collection.
\end{abstract}
\newpage
\section{Introduction}
\label{sec:intro}

Consider the following algorithmic problem: given as input an infinite stream of strings from an unknown target language (one of a known collection of languages), learn to generate new and previously unseen strings also belonging to this target language in the limit. This problem was recently formalized in the work of \citet{kleinberg2024language} with a view to synthesize the core problem at the heart of large language models. The same problem, but with the goal of \textit{identifying} the target language in the collection instead of simply generating from it, has been extensively studied in classical work on language identification in the limit \citep{gold1967language, angluin1979finding, angluin1980inductive}. In fact, we have a precise characterization \citep{angluin1979finding, angluin1980inductive} of the collections of languages for which this problem is tractable. As it turns out, for essentially any interesting collection of formal languages (any collection at all that is not finite, e.g., regular, context-free, context-sensitive), the language identification problem is intractable, not just in the amount of computational power that an algorithm might require, but in a strict computability sense---any algorithm can be fed inputs consistent with a target language from the collection in such a way that the algorithm makes an incorrect guess of the target language infinitely often.

Given such strong negative results for language identification, 
\citet{kleinberg2024language} showed a remarkable positive result for language \textit{generation}. 
They showed that language generation is
tractable for \textit{every} countable collection of languages! Their work gives a simple and elegant algorithm, which, given any stream of input strings from any target language $K$ in a countable collection $\mcC=\{L_1,L_2,\ldots\}$, generates a sequence of previously unseen strings such that beyond a finite time step $t$, the generated strings all belong to the unknown language $K$.
Furthermore, they show that this algorithm can be implemented with only membership query oracle access to the collection $\mcC$. Namely, the algorithm simply needs to be able to ask queries of the form ``does $w$ belong to $L_i$?'' for any string $w$ in the universe and language $L_i \in \mcC$. %

Given that this is possible, it is natural to ask for more. %
How quickly can we hope to achieve this guarantee of generation in the limit?
In fact, the work of \citet{kleinberg2024language} also shows an algorithm that achieves a stronger guarantee for generation in the limit from \textit{finite} collections of languages: as soon as the algorithm sees $t^*$ many distinct strings,
where $t^*$ depends neither on the target language nor its enumeration order, but is only a function of the collection $\mcC$, the algorithm correctly generates from the target language. This algorithm is quite different from their all-purpose algorithm which achieves generation in the limit more generally for all countable collections. Inspired by this, the very recent work of \citet{raman2024generation} seeks to obtain a precise characterization of the collections of languages (beyond just finite collections) for which such guarantees can be obtained. Adopting their terminology, the collections for which $t^*$ may depend only on the target language, but not on its enumeration order, are said to be \textit{non-uniformly generatable}, whereas the collections for which $t^*$ is a function of the language collection, and depends neither on the target language nor its enumeration order, are said to be \textit{uniformly generatable}. \citet{raman2024generation} show that the collections that are uniformly generatable are exactly those that have a bounded complexity measure, termed as the \textit{closure dimension}. They leave the characterization of collections that are non-uniformly generatable largely open.

A different line of inquiry, also motivated from the work of \citet{kleinberg2024language}, stems from a tradeoff between \textit{validity} and \textit{breadth} which their (general-purpose) generation algorithm has to incur. Concretely, the algorithm starts off generating invalid strings, then steadily refines its hallucinations, before eventually settling to generate from an increasingly small subset of the language. Thus, en route to becoming a valid generator, it appears that the algorithm has to sacrifice on generating the entire bulk of the target language. This phenomenon also notoriously shows up in the form of \textit{mode collapse} while training generative adversarial networks \citep{arjovsky2017towards, arjovsky2017wasserstein}. As an open direction in their work, \citet{kleinberg2024language} asked if such a tradeoff is provably necessary for achieving generation in the limit. 

\subsection{Overview of Results}
\label{sec:overview}

As our first result, we show that it is possible to \textit{non-uniformly} generate in the limit from every countable language collection $\mcC=\{L_1,L_2,\ldots\}$. We give a simple algorithm (\Cref{sec:non-uniform-generation-ub}), which has the property that it generates a valid, unseen string from the target language being enumerated as input, as soon as it sees a fixed constant number of distinct strings in the enumeration, where this constant depends \textit{only} on the target language and the collection, but \textit{not} the enumeration order. We note that \citet{kleinberg2024language}'s algorithm for generation in the limit does not satisfy this latter property---namely, for adversarial enumerations of the target language, their algorithm could take arbitrarily long before starting to validly generate. Thus, in addition to being a different algorithm than theirs for generation in the limit, our algorithm satisfies a stronger property.

\begin{theorem}[Non-uniform Generation for Countable Collections]
    \label{thm:non-uniform-generation-ub-overview}
    There exists an algorithm that non-uniformly generates in the limit from every countable collection of languages.
\end{theorem}

\Cref{thm:non-uniform-generation-ub-overview} also answers an open question raised by \citet{raman2024generation}, who asked if every countable collection can be non-uniformly generated in the limit, in the affirmative. The language-dependent constant for the number of distinct strings that our non-uniform generation algorithm must see is formalized in terms of a \textit{non-uniform complexity} measure of the language (see \Cref{def:m(L)} and \Cref{thm:non-uniform-ub}).

However, despite not being a non-uniform generation algorithm, the algorithm of \citet{kleinberg2024language} has the attractive property that it can be implemented using only \textit{membership queries}. Namely, their generation algorithm can be implemented given access to an oracle that can answer queries of the form ``$w \in L_i$?'' for any string $w$ and language $L_i \in \mcC$ of the algorithm's choosing. Given this, one may ask if our non-uniform generation algorithm from above can also be implemented with sole access to such a membership query oracle. Our next result provides a strong negative answer to this question.

\begin{theorem}[Non-uniform Generation Membership Query Lower Bound]
    Any algorithm that non-uniformly generates from all countable collections cannot be solely implemented with membership queries.
\end{theorem}

Our lower bound in fact shows that a generation algorithm that only issues membership queries cannot even non-uniformly generate from all \textit{finite} (even size $2$!) collections. This result shows that non-uniform generation provably requires more computational power in the form of different and stronger oracles. We note that our result does not contradict the closure-based uniform generation algorithm of \citet{kleinberg2024language}, which can be implemented using membership queries. The key detail is that their algorithm starts off with a small piece of additional information about the collection; perhaps surprisingly, we find that this information is entirely crucial (see \Cref{thm:membership-query-lb} and the discussion at the end of \Cref{sec:mem-query-lb}).

Our next result addresses the open question in \citep{kleinberg2024language} regarding the tradeoff between validity and breadth in generation.  Towards this, we propose a definition of \textit{exhaustive} generation in the limit. This definition still requires an algorithm to eventually always generate from the target language (validity). However, there is an additional requirement concerning breadth of generation: beyond some finite time, it ought to be possible to \textit{terminate} the input entirely and stop looking at additional examples. If the algorithm is now asked to continue generating strings indefinitely, it should be the case that the set of strings it would go on to generate, when combined with the input seen so far and the strings it previously generated, cover the entirety of the target language. We formalize these requirements in \Cref{def:exhaustive-generation}, and show a strong negative result for exhaustive generation.

\begin{restatable}[Exhaustive Generation Lower Bound]{theorem}{exhaustivegenerationlb}
    \label{thm:exhaustive-generation-lb}
    There exists a countable collection $\mcC$ (of regular languages) that cannot be exhaustively generated in the limit.
\end{restatable}

\Cref{thm:exhaustive-generation-lb} provides a concrete answer to the open question asked by \citet{kleinberg2024language}, showing that the validity-breadth tradeoff is necessary in a formal sense. 
Our result thus adds to evidence in the literature suggesting that language models with desirable properties must necessarily hallucinate \citep*{kalai2024calibrated, xu2024hallucination, banerjee2024llms, wu2024no}.

Our next contribution is a \textit{precise characterization} of the collections of languages that allow for exhaustive generation, similar to that given by \cite{angluin1980inductive} for identifiability. Angluin introduced two conditions (properties of language collections) in her work, which we denote as ``Angluin's Condition with Enumeration'' (see \eqref{condition:angluin-condition-with-enumeration}) and ``Angluin's Condition with Existence'' (see \eqref{condition:angluin-condition-with-existence}). The condition with enumeration is both a sufficient and necessary condition for identifiability, while the condition with existence is only a necessary condition. These conditions posit the existence of certain \textit{tell-tale sets} $T_i$ for each language $L_i$ in the collection---informally, the set $T_i$ allows an algorithm to distinguish the case when $L_i$ is the target language from cases where the target language is some other language $L_j$.

We show that a \textit{weaker} version of Angluin's Condition with Existence, which we denote as ``Weak Angluin's Condition with Existence'' (see \eqref{condition:weak-angluin-condition-with-existence}), is both necessary and sufficient for exhaustive generation.

\begin{restatable}[Exhaustive Generation Characterization]{theorem}{exhaustivegenerationcharacterization}
    \label{thm:exhaustive-generation-characterization}
    A countable language collection $\mcC$ can be exhaustively generated if and only if it satisfies Weak Angluin's Condition with Existence.
\end{restatable}

Lastly, we consider a setting of uniform language generation, where at each time step, the algorithm 
has the additional ability of asking an oracle whether any string of its choice belongs to the \textit{target} language that is being enumerated. This setting is inspired by one of the many \textit{informant} models considered in the classical work of \citet{gold1967language} for the problem of language identification, and seeks to model real-life learning scenarios where an algorithm is allowed to ask for feedback at every time step (e.g., as in RLHF \citep{christiano2017deep}). Note that this model is different from the membership query model---there, the algorithm can only ask the oracle queries regarding a language of its choice, but it crucially does not know which language is the target language. In particular, we note that in this model, even \textit{identification} becomes possible for a large class of countable collections (see Table 1 in \citet{gold1967language}). %
Furthermore, we know already from the result of \citet{kleinberg2024language} that all countable collections can be generated in the limit (without requiring any feedback), and from our \Cref{thm:non-uniform-generation-ub-overview} above, that they can also be non-uniformly generated in the limit (also, without any feedback). Thus, the interesting question here is to study whether collections that cannot be uniformly generated without feedback become uniformly generatable, once we allow feedback---as \Cref{eg:gf-but-no-gnf} illustrates, this is indeed true.

Similar to how \citet{raman2024generation} characterize uniformly generatable languages (without feedback) in terms of a property of the collection, we completely characterize (\Cref{sec:genfeed}) the collections of languages that can be uniformly generated in this feedback model, in terms of an abstract complexity measure which we term the Generation-Feedback ($\gf$) dimension.

\begin{theorem}[Uniform Generation with Feedback]
    \label{thm:genfeed}
    A collection $\mcC$ of languages can be uniformly generated with feedback if and only if its $\gf$ dimension is finite.
\end{theorem}

We supplement our results with interesting examples of language collections throughout the paper that provide intuition and elicit differences between different models.

\subsection{Related Work}
\label{sec:related-work}

Our work directly builds off of the %
thought-provoking work of \citet{kleinberg2024language}, which, in addition to providing a formal model and positive results for language generation in the limit, also provides interesting results in the context of prompt completion in language models. While the work of \citet{kleinberg2024language} is also the primary motivating work for \citet{raman2024generation}, our results on non-uniform generation are based on the definitions of uniform and non-uniform generation given by \citet{raman2024generation}, and the corresponding open question raised by them regarding the characterization of non-uniform generation. %

\paragraph{Concurrent and Independent Work.}
We discuss the relationship of our paper with two different parallel and independent works on language generation. Our results on non-uniform generation resolve an open problem from one work \citep{raman2024generation}, and independently of us, the authors also obtained similar results in an updated version of that paper. On the other hand, our results (in the earlier version of this manuscript \citep{charikar2024exploring}) on the tradeoff between validity and breadth were independent of the work of a different set of authors \citep*{kalavasis2024limits}. In this case, both we and they independently realized that we could build on the results in our manuscript to obtain tight characterizations. We elaborate on this below.

\paragraph{Non-uniform generation:} In parallel with, and independent of our work, \citet{raman2024generation} updated their paper \citep*{li2024generation}.
Their updated manuscript includes a result on non-uniform generation for countable collections, resolving an open problem they raised earlier.
The conclusion of Corollary 3.7 in \citet*{li2024generation} is similar to that of \Cref{thm:non-uniform-generation-ub-overview} in the present paper.

\paragraph{Validity-breadth tradeoff:}
In the time that we were preparing the original version of this manuscript \citep{charikar2024exploring}, and independently of our work, the very recent paper of \citet*{kalavasis2024limits} also considers the validity-breadth tradeoff arising in \citet{kleinberg2024language}. In particular, \citet{kalavasis2024limits} also formalize a definition of \textit{generation with breadth}, and obtain negative results similar to ours, providing evidence that the tradeoff between validity and breadth is necessary in a sense. While our definition of exhaustive generation is similar to their definition of generation with breadth in that both seek to formalize not missing out on any strings in the target language, there are certain important differences. In particular, a generator that satisfies their definition also satisfies our definition, but not vice versa---in this sense, ours is a \textit{weaker} requirement of breadth than theirs. For this reason, our negative result for exhaustive generation allows us to also answer an open question asked by \citet{kalavasis2024limits} (see \Cref{sec:consistency-breadth-impossible}). We elaborate further on the connections to the work of \citet{kalavasis2024limits} in \Cref{sec:connections-to-breadth}. We note that the work of \citet{kalavasis2024limits} also extensively studies the \textit{rates} at which generation in the limit can be achieved, within the framework of universal learning due to \citet{bousquet2021theory}. Here, instead of an online setting where the target language may be arbitrarily enumerated in a worst-case fashion, they consider a statistical setting where the stream of input strings is drawn i.i.d. from a \textit{distribution} supported on the target language, and the quantity under consideration is the rate at which validity/breadth of generation increases, with the number of input samples. Finally, we note that certain other definitions of generation with breadth (see \Cref{remark:kmv-other-notions} ahead) are also considered by \cite{kalavasis2024limits}.

\paragraph{Characterizations of the validity-breadth tradeoff:} Shortly after our initial manuscript appeared, the authors of \citep{kalavasis2024limits} were able to build on our results in \cite{charikar2024exploring} and their work to obtain a tight characterization of languages that satisfy various definitions of generation with breadth (like exact breadth, approximate breadth, etc.), as well as exhaustive generation considered in our works. Their new manuscript \citep{kalavasis2024characterizations} has several results that mirror results that we obtained in parallel and are presented in this updated version of our manuscript.
After we learned of their results from personal communication, we coordinated with them to make our definitions and terminology consistent. We provide a more detailed comparison between our results and theirs in \Cref{sec:kmv-characterizations-comparison}.

\section{Preliminaries}

We largely follow the problem setup of \cite{kleinberg2024language}, who build upon the setup originally introduced in \cite{gold1967language}. Let $\Sigma$ be a finite alphabet set (e.g., $\{0,1\}, \{a,b,\ldots,z\}$), and let $\Sigma^*$ denote the set of all strings of finite length formed by concatenating elements from $\Sigma$ in any order. A language $L$ is simply a countable subset of $\Sigma^*$, and we will always be concerned with generating from a \underline{countable} collection $\mcC=\{L_1,L_2,\ldots\}$ of languages. For generation in the limit to make sense, we will assume that $|L_i|=\infty$ for every $i$. The set of all integers is denoted by $\Z$.

\subsection{Generation in the Limit}

The setup assumes that a target language $K \in \mcC$ is chosen and fixed, and thereafter, strings from $K$ are provided sequentially as input in the form of an enumeration $x_1,x_2,x_3,\ldots$. The choice of $K$ and the order of enumeration can possibly be chosen by an adversary. Repetitions of strings are permitted in the enumeration---the only requirement is that every string in $K$ appears at least once in the enumeration, i.e., for every $z \in K$, $z=x_t$ for some $t < \infty$. \textbf{We use the shorthand $S_t$ to denote the set of all distinct strings seen in an enumeration $x_1,\ldots,x_t$ up until $t$.}

\begin{definition}[Generation in the limit \citep{kleinberg2024language}]
    \label{def:generation-in-the-limit}
    An algorithm generates in the limit from languages in a collection $\mcC$, if for any language $K \in \mcC$ and any enumeration of $K$ presented to the algorithm, there exists $t^* < \infty$ such that for all $t \ge t^*$, the string $z_t$ generated by the algorithm at time step $t$ belongs to $K \setminus S_t$.
\end{definition}

\begin{remark}
    Note that the time $t^*$ in the definition above above can depend on both the target language $K$ as well as the particular enumeration of it.
\end{remark}

In the above definition, the only requirement of the algorithm is that it is a \textit{computable map} from the input seen so far to an output string (\cite{kleinberg2024language} refer to this setting as ``generation in the limit via a function'').
We will be explicit when we care about the computational power required by the algorithm to compute this map. For example, the way in which \cite{kleinberg2024language} account for the computational power required by a generating algorithm is via the \textit{membership query} model.
Here, at each step, an algorithm is allowed to make finitely many queries to an oracle of the form ``$w \in L_i$?'', for any $w \in \Sigma^*$ and any $L_i \in \mcC$ of its choice. Notably, the algorithm cannot make the query  ``$w \in K$?'' corresponding to the unknown target language $K$. 
\cite{kleinberg2024language} showed that there exists an algorithm that successfully generates in the limit from any language in a countable collection using only membership queries.

\subsection{Uniform/Non-uniform Generation}
\label{sec:prelims-uniform-non-uniform-gen}

In the case that the collection $\mcC$ is finite, \cite{kleinberg2024language} additionally showed that it is possibly to construct an algorithm that \textit{uniformly} generates from languages in $\mcC$: namely, as soon as the algorithm sees $t^*=t^*(\mcC)$ many distinct strings from any $K \in \mcC$, \textit{irrespective} of $K$ and its enumeration order, it generates from $K \setminus S_t$ successfully for every $t \ge t^*$. Inspired by this, the recent work of \cite{raman2024generation} formalizes the following distinctions of ``non-uniform'' and ``uniform'' generation of a language.

\begin{definition}[Non-uniform generation in the limit (Definition 3 in \cite{raman2024generation})]
    \label{def:non-uniform-generation-by-function}
    An algorithm non-uniformly generates in the limit from languages in a collection $\mcC$, if for any language $K \in \mcC$, there exists a $t^*=t^*(\mcC, K)$ such that for any enumeration of $K$ presented to the algorithm, the string $z_t$ generated by the algorithm at time step $t$ belongs to $K \setminus S_t$ for all $t$ satisfying $|S_t| \ge t^*$.
\end{definition}

\begin{definition}[Uniform generation in the limit (Definition 4 in \cite{raman2024generation})]
    \label{def:uniform-generation-by-function}
    An algorithm uniformly generates in the limit from languages in a collection $\mcC$, if there exists a $t^*=t^*(\mcC)$ such that for any language $K \in \mcC$ and any enumeration of $K$ presented to the algorithm, the string $z_t$ generated by the algorithm at time step $t$ belongs to $K \setminus S_t$ for all $t$ satisfying $|S_t| \ge t^*$.
\end{definition}

\cite{raman2024generation} generalize the uniform generation result of \cite{kleinberg2024language} for finite collections to possibly infinite collections by showing that any collection $\mcC$ of languages having \textit{bounded complexity} (defined in terms of having a finite ``closure dimension'') admits uniform generation. Furthermore, if a collection $\mcC$ has infinite closure dimension, then it is not possible to uniformly generate from languages in $\mcC$.

\begin{definition}[Closure dimension]
\label{def:closure-dimension}
    The closure dimension of a collection $\mcC$ of languages is the size of the largest set $S=\{x_1,\ldots,x_d\}$, such that the intersection of all languages in $\mcC$ that contain $S$ is finite.
\end{definition}

\subsection{Exhaustive Generation}
\label{sec:prelims-exhaustive-generation}

The generation algorithm of \cite{kleinberg2024language} exhibits a tension between \textit{validity} of outputs and \textit{breadth} of generation, as also explicitly stated by them.
Namely, the algorithm starts off by generating strings that could possibly not belong to the target language $K$ for a while, before eventually settling to generate from subsets of $K$ that seemingly get smaller and smaller. \cite{kleinberg2024language} leave the problem of bridging this gap open, asking if it is possible to construct an algorithm that generates from $K$ with breadth (i.e., does not miss out on generating any strings from $K$), or if there is a formal sense in which such a tradeoff is necessary. To model this tradeoff, we propose a definition of \textit{exhaustive} generation. 

For this, we consider a generating algorithm $\mcA$, which at any time $t$, maps 
$S_t$ to a generator $\mcG_t:\N \to \Sigma^*$. We can imagine that the string generated by the algorithm at time step $t$ is simply $\mcG_t(1)$. However, we want to also consider what happens if the input were to be \textit{terminated} beyond time step $t$. In this case, we want to study the sequence of strings $\mcG_t(1), \mcG_t(2), \mcG_t(3),\ldots$ that would be generated by $\mcG_t$---we can think of this latter scenario to be a form of ``generate-only'' mode that is implicitly defined by the generator $\mcG_t$. We use the shorthand $Z_{< t}$ to denote the set of distinct strings in the sequence $\mcG_1(1), \mcG_2(1),\ldots,\mcG_{t-1}(1)$, and the shorthand $Z_{\ge t}$ to denote the set of distinct strings in the sequence $\mcG_t(1), \mcG_t(2), \ldots$ generated by $\mcG_t$ were it to go into generate-only mode from time $t$.

\begin{definition}[Exhaustive Generation]
    \label{def:exhaustive-generation}
    A generating algorithm $\mcA$ exhaustively generates in the limit from languages in a collection $\mcC$, if for any $K \in \mcC$ and any enumeration of $K$, there exists $t^* < \infty$ such that for any $t \ge t^*$, it holds that
    \begin{enumerate}
        \item $|Z_{\ge t} \setminus K| < \infty$.\footnote{In a previous version of our manuscript, we stated this condition as $Z_{\ge t} \subseteq K$, i.e., the generator makes no mistakes after it goes into generate-only mode. The updated condition allows the generator to make a finite number of mistakes after it goes into generate-only mode. We believe that the updated condition is a more natural definition for exhaustive generation, in keeping with the spirit of making no mistakes for language generation in the limit.}
        \item $S_{t} \cup Z_{< t} \cup Z_{\ge t} \supseteq K$.\footnotemark
    \end{enumerate}
    \footnotetext{We could have alternatively considered this condition to be $Z_{< t} \cup Z_{\ge t} \supseteq K$ instead, which implies the written condition, since $S_t \subseteq K$. Conversely, a generator that relies on $S_t$ to satisfy condition 2 in \Cref{def:exhaustive-generation} could simply generate $S_t$ first, before going on to generate $Z_{\ge t}$ which it was planning to generate.}
\end{definition}

We note that unlike uniform/non-uniform generation considered in \Cref{sec:prelims-uniform-non-uniform-gen}, $t^*$ in the above definition is allowed to depend on both $K$ and its enumeration order, just as in \Cref{def:generation-in-the-limit}.

The first condition in the definition above is in line with the requirement that the generator, if asked to go into generate-only mode, can only generate finitely many spurious strings, and thereafter, completely stops hallucinating. The second condition ensures coverage of the entire language. Since we hope that the generator will largely produce new and previously unseen examples, our definition effectively throws in the examples already presented in covering the target language $K$ and does not penalize the generator for not regenerating those.
We note that in \Cref{remark:relaxed-exhaustive-generation} in \Cref{sec:exhaustive-generation-characterization}, we also comment about a slightly relaxed version of \Cref{def:exhaustive-generation}.

The notion of such a ``generate-only'' mode may seem unnatural at first, but in fact, it is quite natural to terminate the input and not look at further examples in the context of exhaustive generation. Recall that the input is a valid enumeration of the target language $K$; that is, every string in $K$ eventually appears in it.
Therefore, it is natural to terminate the input at some point and consider the output of the generator in ``generate-only'' mode.

\section{Non-uniform Language Generation for a Countable Collection}
\label{sec:non-uniform-generation-ub}

In this section, we show that every countable collection $\mcC$ of languages can be non-uniformly generated. We consider $\mcC$ to be specified in a given enumeration $\mcC=\{L_1,L_2,L_3,\ldots\}$. %

\paragraph{Algorithm.} 
Consider the algorithm, which at step $t$ in the enumeration of its input, %
initializes $I_t = \Sigma^*$, and iterates through the languages $L_1,\ldots,L_t$ in this order. Whenever it encounters a language $L_i$ that contains $S_t$, it checks if $|I_t \cap L_i|=\infty$. %
If it is, then it updates $I_t$ as $I_t=I_t \cap L_i$. Otherwise, it skips over $L_i$, and leaves $I_t$ unaffected. Thus, throughout the algorithm's iteration over $L_1,\ldots,L_t$, the following invariants are maintained: (1) $I_t$ is an infinite set, and (2) $I_t$ is the intersection of a finite set of languages that contain $S_t$.\footnote{We use the convention that the intersection of zero languages is $\Sigma^*$.} The algorithm then generates an arbitrary string from $I_t \setminus S_t$\footnote{If we do not want repetitions in the strings generated by the algorithm, we can have it generate an arbitrary ungenerated-as-yet string from the infinite $I_t \setminus S_t$.}.

We will now show that the above algorithm non-uniformly generates from languages in $\mcC$. We will specify the non-uniform guarantee of the algorithm in terms of the \textit{non-uniform complexity} of languages in the collection $\mcC$.

\begin{definition}[Non-uniform Complexity]
    \label{def:m(L)}
    Given $\mcC=\{L_1,L_2,L_3,\ldots$\}, for any $i \in \N$, define the non-uniform complexity $m_\mcC(L_i)$ of $L_i$ as
    \begin{align}
        \label{eqn:def-m(L)}
        m_\mcC(L_i) := \max  \left\{ \left| \bigcap_{L \in \mcC'}L \right|: \mcC' \subseteq \{L_1,\dots,L_i\}, \mcC' \ni L_i, \left| \bigcap_{L \in \mcC'}L \right| < \infty \right\}.
    \end{align}
\end{definition}

\begin{theorem}
    \label{thm:non-uniform-ub}
    For any language $L_{i^*} \in \mcC$, and any enumeration of $L_{i^*}$ presented as input to the above algorithm, the algorithm generates from $L_{i^*} \setminus S_t$ for all $t$ satisfying 
    \begin{align}
        \label{eqn:non-uniform-timer-value}
        |S_t| \ge \max(i^*, m_\mcC(L_i^*) + 1). 
    \end{align}
\end{theorem}

\begin{proof}

    Consider any $t$ where $|S_t| \ge \max(i^*, m_\mcC(L_i^*) + 1)$. Note that such a $t$ necessarily satisfies $t \ge i^*$, and hence, the target language $L_{i^*}$ is under consideration by the algorithm at this step. Furthermore, $L_{i^*}$ definitely contains $S_t$. So, the algorithm proceeds to iterate over $L_1,\ldots,L_t$, maintaining $I_t$ along the way. We only need to argue that when the algorithm encounters $L_{i^*}$ and considers $I_t \cap L_{i*}$, it finds that this set is infinite. If this is true, then $I_t$ would be updated to be $I_t \cap L_{i^*}$, which is a subset of $L_{i^*}$. Thereafter, $I_t$ will only become a smaller subset of $L_{i^*}$ as the algorithm proceeds in its iteration. Finally, because the algorithm maintains $I_t$ to be infinite, at the end of its iteration, $I_t$ is an infinite subset of $L_{i^*}$. Thus, the algorithm can safely generate a string from $I_t \setminus S_t$, as it is guaranteed to also belong to $L_{i^*} \setminus S_t$.

    So, we continue to argue that when the algorithm encounters $L_{i^*}$ and considers $I_t \cap L_{i*}$, it finds that this set is infinite. Suppose this were not the case---then, $I_t \cap L_{i^*}$ is a finite set. Observe however, that $I_t \cap L_{i^*}$ is an intersection of languages that contain $S_t$. Therefore, since $|S_t| \ge m_\mcC(L_i^*) + 1$, we have that $|I_t \cap L_{i^*}| \ge m_\mcC(L_i^*) + 1$. But notice also that $I_t$ is an intersection of languages that appear before $L_{i^*}$ in the enumeration of $\mcC$. 
    Thus, we have obtained a collection of languages that is a subcollection of $\{L_1,\ldots, L_{i^*}\}$, contains $L_{i^*}$, and has a finite intersection of size at least $m_\mcC(L_i^*) + 1$: this contradicts the definition of $m_\mcC(L_i^*)$ (see \eqref{eqn:def-m(L)}). Thus, $I_t \cap L_{i*}$ must be infinite.
\end{proof}

At its heart, our algorithm is closer to the closure-based algorithm of \cite{kleinberg2024language} for finite collections. Namely, while that algorithm considers generating from the intersection of \textit{all} languages in the collection consistent with the input (namely, the closure), our algorithm is more conservative among the languages it considers for computing the closure. Ultimately, it is able to greedily choose these languages, keeping track of a rather simple criterion---that of infinite intersection.

We make a concluding remark that our algorithm above is not collection-specific: it simultaneously works for all countable collections (defined on the fixed universe $\Sigma^*$), just like the generation algorithm of \cite{kleinberg2024language} for countable collections. Furthermore, our algorithm can be implemented with a membership query oracle, \textit{and} an additional oracle, which, when queried with any finite collection of languages, responds with whether the intersection of the languages in the collection is finite or not.

\section{Lower Bound for Non-uniform Language Generation Using Only Membership Queries}
\label{sec:mem-query-lb}

Given that the generation algorithm of \cite{kleinberg2024language} can be implemented using membership queries, it is natural to ask if our non-uniform generation algorithm from above can also be implemented using only membership queries. Towards this, we show a strong negative result---we show that it is impossible (in a computability sense) for any algorithm to (simultaneously) non-uniformly generate for every finite collection, with just membership queries.

Somewhat surprisingly, our lower bound applies to collections with just two languages $L_0$ and $L_1$, when the algorithm has no a priori information about $L_0$ and $L_1$ and can only access the languages using membership queries.
While this result might seem to contradict the results of \cite{kleinberg2024language} who give uniform generation algorithms for finite collections of languages, we reconcile this apparent contradiction in our discussion after establishing the formal theorem.

Informally, the theorem says that we cannot have an algorithm with a non-uniform generation guarantee for every collection of two languages.
A non-uniform guarantee for such a collection $\mcC=\{L_0,L_1\}$ means that, for any input that is a valid enumeration of $L_0$ (respectively $L_1$), there is a bound $t_0$ (respectively $t_1$) independent of the enumeration, such that the algorithm correctly generates from $L_0$ (respectively $L_1$) after time step $t_0$ (respectively $t_1$).
The proof takes a supposed valid algorithm $\mcA$ and constructs a bad input, i.e., a specific collection of two languages where $\mcA$ fails to satisfy the non-uniform generation guarantee.
This adversarial input is one that keeps the algorithm guessing at every step, i.e., the algorithm learns no information to determine whether the target languages is $L_0$ or $L_1$. For any time step $t$ beyond $\max(t_0,t_1)$, the algorithm is supposed to make no mistakes.
However, we can easily force a mistake at time $t$ by picking the target language to be one of $L_0,L_1$, and revealing this to the algorithm after time step $t$.

\begin{theorem}
    \label{thm:membership-query-lb}
    There cannot exist a (deterministic) algorithm $\mcA$ that only makes membership queries, and %
    satisfies the following property: for every collection $\mcC=\{L_0,L_1\}$ of two languages (on a universe $\Sigma^*$), there exist $t^*(\mcC, L_0) < \infty$ and $t^*(\mcC, L_1) < \infty$ such that:
    \begin{enumerate}
        \item For every enumeration $x_1,x_2,\dots$ of %
        $L_0$ and every $1 \le t < \infty$, the algorithm makes a finite number of queries at step $t$, and if $|S_t| \ge t^*(\mcC, L_0)$, then $\mcA(S_t)\footnotemark \in L_0 \setminus S_t$.
        \item For every enumeration $x_1,x_2,\dots$ of %
        $L_1$ and every $1 \le t < \infty$, the algorithm makes a finite number of queries at step $t$, and if $|S_t| \ge t^*(\mcC, L_1)$, then $\mcA(S_t) \in L_1 \setminus S_t$.
    \end{enumerate}
    \footnotetext{Here, we denote by $\mcA(S_t)$ the string $z_t$ generated by the algorithm $\mcA$ at time step $t$.}
\end{theorem}
\begin{proof}
    Assume for the sake of contradiction that $\mcA$ is a valid generating algorithm satisfying the property. We will adversarially construct two languages $L_0=L_0(\mcA)$ and $L_1=L_1(\mcA)$ such that $\mcA$ does not satisfy the property on the collection $\mcC=\{L_0,L_1\}$.

    We will build up $L_0$ and $L_1$ in phases $\mcP_1,\mcP_2,\dots$, as a function of the execution of $\mcA$. The first time that we will consider any string $w \in \Sigma^*$, we will irrevocably decide whether $w$ belongs to exactly one of $L_0$ or $L_1$, or both of them. Towards this, we will maintain a global variable $a$ whose state is tracked across phases: it is initialized to $a=0$ before the start of Phase $\mcP_1$, and %
    maintains its value
    from the end of phase $\mcP_{m-1}$ to the start of phase $\mcP_m$. We also maintain a dictionary $\assigned$, whose keys are strings in $\Sigma^*$, and initialize $\assigned[w]=-1$ for every string $w \in \Sigma^*$. Each phase $\mcP_m$ operates as follows:
    \begin{framed}
        \begin{align*}
            &\text{\underline{Description of Phase $\mcP_m$}:} \\
            &\text{(1) //} \blue{\textit{prepare next string to be enumerated}}\\
            &\qquad \text{Insert a yet unseen, distinct string $x_m$ in both $L_0$ and $L_1$,} \\
            &\qquad \text{and set $\assigned[x_m]=\{0,1\}$.} \\
            &\text{(2) Feed $x_m$ as the next input to $\mcA$.} \\
            &\text{(3) //} \blue{\textit{process membership queries made by the algorithm}} \\
            &\qquad \text{Initialize $j=1$, and loop until $\mcA$ generates some $z_m$}: \\
            &\qquad\qquad \text{Suppose $\mcA$ queries ``$y_{j} \in L_{a_j}$?'' for $a_j \in \{0,1\}$.} \\
            &\qquad\qquad \text{If $\assigned[y_{j}] \neq -1$:} \\
            &\qquad \qquad\qquad // \blue{\textit{$y_{j}$ has already been placed in $L_0,L_1$ earlier}} \\
            &\qquad \qquad\qquad \text{Answer ``Yes''/``No'' according to $\assigned[y_{j}]$.} \\
            &\qquad\qquad \text{Else:} \\
            &\qquad \qquad\qquad //        
            \blue{\textit{place $y_{j}$ in exactly one of $L_0,L_1$}} \\
            &\qquad\qquad\qquad \text{Insert $y_{j}$ in $L_a$, set $\assigned[y_{j}]=a$.} \\
            &\qquad\qquad\qquad\text{Answer ``Yes''/``No'' based on $a_j\stackrel{?}{=}a$, and thereafter, set $a=1-a$.} \\
            &\text{(4) //} \blue{\textit{process string guessed by the algorithm}}\\ 
            &\qquad \text{If $\assigned[z_m] = -1$ for the generated string $z_m$:}\\
            &\qquad\qquad \text{Insert $z_{m}$ in $L_a$, set $\assigned[z_m]=a$, and thereafter, set $a=1-a$.}
        \end{align*}
    \end{framed}
    This concludes the construction of the languages $L_0$ and $L_1$. Now, suppose there were some finite $t^*(\mcC, L_0)$ and $t^*(\mcC, L_1)$ for which the generations of $\mcA$ satisfied the property in the theorem statement. Let $n = \max(t^*(\mcC, L_0), t^*(\mcC, L_1))$. We note a few key properties of our construction that we exploit in the rest of the proof: (i) $L_0$ and $L_1$ are infinite languages. (ii) Each string enumerated in Step (2) as input belongs to both $L_0$ and $L_1$, thus the enumeration $x_1, x_2, \ldots x_m$ (for $m \leq n$) is valid for both $L_0$ and $L_1$. (iii) The adversary can force the algorithm to make a mistake for the generated string $z_n$ (and hence violate the correctness guarantee) by committing to one of $L_0$, $L_1$ as the true language.

    Consider feeding the generator the input sequence $x_1,x_2,\dots, x_n$, where for $m \in [n]$, $x_m$ is the string from Step (1) in Phase $\mcP_m$ above. 
    Note that the number of distinct strings $|S_n|=n$. Suppose that $z_1,\dots,z_n$ are the strings $\mcA$ generates at each step. Observe that by design, upon feeding $x_m$, $\mcA$ asks exactly the same sequence of queries that we considered in the loop in Step (3) above. 
    
    We first argue that $\mcA$ must only ask a finite number of queries, and then generate $z_m$, which is exactly the string generated by $\mcA$ at the culmination of the loop in Step (3). If this is not the case, then the loop in Step (3) would never have terminated. The infinite loop ensures that both $L_0$ and $L_1$ continue being populated to be infinite languages---hence, we have a legal collection $\mcC$. 
    Note that the sequence $x_1, x_2, \ldots x_m$ produced so far can be extended to be a valid enumeration of either $L_0$ or $L_1$ (since each of $x_1,\dots,x_m$ belongs to both $L_0$ and $L_1$, as ensured by Step (1)).
    We now simply declare $L_0$ to be the true (target) language. 
    Since the sequence $x_1, x_2, \ldots x_m$ is the prefix of a valid enumeration of $L_0$, $\mcA$ must satisfy the correctness guarantee for a valid generation algorithm and cannot loop forever.

    So, consider the string $z_n$ generated by $\mcA$ after seeing the last input $x_n$---note that by the correctness guarantee, $z_n$ must not belong to $S_n = \{x_1,x_2,\ldots x_n\}$. Consider the value of $\assigned[z_n]$ after Step (4) in Phase $\mcP_n$; suppose $\assigned[z_n]=0$. We will then declare $L_1$ to be the true language, and thereafter, simply begin re-enumerating all of $L_1$ afresh. Similarly, if $\assigned[z_n]=1$, we will declare $L_0$ to be the true language, and thereafter, simply begin re-enumerating all of $L_0$. Notice that either way, we produce a legal and complete enumeration of the true language. More importantly, notice that in either case, the string $z_n$ that $\mcA$ generates at step $n$ \textbf{does not} belong to the true language---this is because $z_n$ is only contained in one of $L_0$ or $L_1$ after Step (4). This violates the guarantee required of $\mcA$, and hence $\mcA$ is not a valid generating algorithm.

\end{proof}

Observe that for finite collections, uniform and non-uniform generation become equivalent, as one can simply assume the uniform bound to be the maximum non-uniform bound across all the languages. In that case, our lower bound above would appear to contradict the closure algorithm of \cite{kleinberg2024language} for finite collections, which uses only membership queries. We now explain why this is not a contradiction.

One way to think about the proof above is that there are two very different strategies for the algorithm to generate from the target language, depending on whether the intersection of $L_0$ and $L_1$ is finite or infinite. For finite intersection of size $d$, from time step $d+1$ onwards, only one of the two languages is consistent with the input since both the languages cannot contain the first $d+1$ positive examples provided to the algorithm. 
Thus, the algorithm will have correctly identified the index of the target language at this point and can easily generate new strings from this language by performing membership queries until it finds a string in the language.
On the other hand, if the intersection of the languages $L_0$ and $L_1$ is infinite, a very different strategy works: we simply perform membership queries until we find a new string $x$ which belongs to both $L_0$ and $L_1$.
Thus, successful strategies in the two cases of finite/infinite intersection both involve running a potentially infinite loop which is guaranteed to terminate given a crucial piece of information, i.e., whether the intersection is finite or not.
In the absence of this information a priori, we show that the algorithm must either loop forever, or can be forced to make a mistake at time step $t$ for arbitrarily large $t$.

Now it is clear why our result does not contradict the uniform generation algorithm of \cite{kleinberg2024language} for finite collections.
The algorithm of \cite{kleinberg2024language} crucially assumes \textit{a priori} knowledge of whether the languages in the finite collection have infinite intersection\footnote{which is a reasonable assumption, since the problem of checking whether the intersection of two context-free languages is infinite is already undecidable.}, and employs very different strategies in the two cases of infinite/finite intersection. In the case of infinite intersection, the algorithm does not wait to see any input at all, and simply queries whether strings in the universe belong to all languages---it is guaranteed to always find such a string in finite time, and it can go on to generate it. Otherwise, if the languages in the collection do not have infinite intersection, the algorithm is additionally told a quantity $t(\mcC)$, which measures the the size of the largest finite intersection of a subcollection---the algorithm then waits till it sees $t(C)+1$ distinct strings (behaving arbitrarily until then). Only at such a time does the algorithm, using membership queries, begin determining if there is a new string beyond the input seen so far that belongs to the the intersection of all languages consistent with the input, i.e., the \textit{closure} (which is guaranteed to be infinite). In fact, if the algorithm did not wait till it sees $t(C)+1$ strings, it is very well possible that the closure is \textit{exactly} equal to the input. In such a case, the algorithm will ceaselessly issue membership queries in its quest to find a new string in the closure beyond the input, and never terminate.

We conclude this section with one final insight.
The construction of the language collection $\mcC=\{L_0,L_1\}$ in the proof above has the following property: at some point in the construction, if $\mcA$ loops forever in the course of making membership queries, where it does not discover any new string that belongs to both languages, then the languages $L_0,L_1$ have finite intersection.
On the other hand, if $\mcA$ always generates a new guess after making finitely many membership queries, where it does not discover any new string that belongs to both languages, then the intersection of languages $L_0,L_1$ is infinite.
This lower bound is easily circumvented by an algorithm that has an a priori promise that the size of the intersection is either finite or infinite, and
our proof crucially exploits the fact that the generation algorithm does not have this a priori information.

\section{Results on Exhaustive Generation}
\label{sec:exhaustive-generation-results}

In this section, we study the setting of exhaustive generation introduced in \Cref{sec:prelims-exhaustive-generation}. First, in \Cref{sec:exhaustive-generation-lb}, we show a lower bound for exhaustive generation. The collection of languages that realizes this lower bound turns to not be identifiable as well, possibly suggesting that identifiability and exhaustive generation are equivalent. However, in \Cref{sec:identifiability-exhaustive-generation-separation}, we exhibit a language collection that is not identifiable in the limit, but \textit{can} be exhaustive generated. Next, in \Cref{sec:connections-to-breadth}, we elaborate on the connections between our notion of exhaustive generation and the notion of generation with breadth introduced by \cite{kalavasis2024limits}. As we argue in \Cref{sec:consistency-breadth-impossible}, our lower bound from \Cref{sec:exhaustive-generation-lb} resolves an open question asked by \cite{kalavasis2024limits}. In \Cref{sec:distinction-breadth-exhaustive-generation}, we show that the collection from \Cref{sec:identifiability-exhaustive-generation-separation} which separates identifiability and exhaustive generation can also \textit{not} be generated with breadth. Recall that this collection is not identifiable, but can be exhaustively generated. This may suggest that amongst exhaustive generation and generation with breadth, the latter is closer to identifiability. Indeed, as we show in \Cref{sec:generation-with-breadth-necessary-condition}, a necessary condition for identifiability (``Angluin's Condition with Existence'') turns out to also be a necessary condition for generation with breadth.

\subsection{Lower Bound for Exhaustive Generation}
\label{sec:exhaustive-generation-lb}

Recall that in exhaustive generation, we are concerned with generating algorithms $\mcA$ that output a generator $\mcG_t:\N \to \Sigma^*$ at every time step $t$. We denote by $Z_{< t}$ the set of distinct strings generated by the algorithm up until time $t-1$ (namely $\mcG_1(1),\ldots,\mcG_{t-1}(1)$), and by $Z_{\ge t}$ the set of distinct strings generated by the algorithm from time $t$ onwards, as if it were to stop seeing any more input (namely $\mcG_{t}(1), \mcG_t(2),\ldots$). As given in \Cref{def:exhaustive-generation}, we desire that for any language $K \in \mcC$ and enumeration of it, there exists a finite $t^* < \infty$ such that for every $t \ge t^*$, it holds that $|Z_{\ge t} \setminus K| < \infty$, and also that $S_t \cup Z_{< t} \cup Z_{\ge t} \supseteq K$.

We restate \Cref{thm:exhaustive-generation-lb}, which shows that even simple language collections cannot be exhaustively generated.

\exhaustivegenerationlb*

\begin{proof}
    Consider the collection $\mcC = L_\infty \cup \bigcup_{i \in \Z}L_i$, where
    \begin{align*}
        &L_\infty = \Z, \\
        &L_i = \{-i, -i+1, -i+2 ,\ldots\} \text{ for $i \in \Z$}.
    \end{align*}
    Each language above is an arithmetic progression, and hence a regular language (e.g., when the strings are expressed in binary representation).
    Assume for the sake of contradiction that there exists an algorithm $\mcA$ that exhaustively generates in the limit from languages in $\mcC$. This means that for any $K \in \mcC$ and any enumeration of $K$, there exists a $t^* < \infty$ such that for any $t \ge t^*$, it holds that (1) $|Z_{\ge t} \setminus K| < \infty$, and (2) $S_t \cup Z_{<t} \cup Z_{\ge t} \supseteq K$. %

    Suppose that an adversary enumerates $L_0$ as $0,1,2,3,\ldots$. 
    Then, there must exist some time step $t^*_0 < \infty$, such that for the generator $\mcG_{t^*_0}$ output by $\mcA$ at $t^*_0$ (which is solely a function of $S_{t^*_0}=\{0,1,2,\ldots,t^*_0\}$, we have that $|Z_{\ge t^*_0} \setminus L_0|<\infty$. Let $t_0=t^*_0$.

    Next, consider the enumeration of $L_{1}$, given as $0,1,2,\ldots,t_0,-1,0,1,2,3,\ldots$. This is a valid enumeration of $L_{1}$, and hence there must exist a finite time step $t^*_{1} < \infty$, such that for any $t \ge t^*_{1}$, the generator $\mcG_{t}$ output by $\mcA$ at $t$ satisfies that $|Z_{\ge t} \setminus L_{1}|<\infty$. 
    Let $t_1 = \max(t_{1}^*, t_0+1)$, and observe that in this enumeration of $L_1$, $t_1$ appears at a timestep beyond $t^*_1$.

    Now, consider the enumeration of $L_{2}$, given as $0,1,2,\ldots,t_0,-1,0,1,2,3,\ldots,t_{1}, -2,-1,0,\ldots$. This is a valid enumeration of $L_{2}$, and hence there must exist a finite $t^*_{2} < \infty$, such that for any $t \ge t^*_{2}$, the generator $\mcG_{t}$ output by $\mcA$ at $t$ satisfies that $|Z_{\ge t} \setminus L_{2}| < \infty$. Let $t_2 = \max(t_{2}^*, t_1+1)$, and observe again that in this enumeration of $L_2$, $t_2$ appears at a time step beyond $t^*_2$.

    We can repeat the same argument indefinitely, which results in the enumeration 
    \begin{align*}
    0,1,2,\ldots,t_0,-1,0,1,2,\ldots,t_{1}, -2,-1,\ldots,t_2, -3,-2,\ldots,t_3,-4,-3,,\ldots
    \end{align*}
    Observe that this is a valid enumeration of $L_\infty$: starting with $i=0$, phase $i$ of the above argument comprises of the sequence $-i,-i+1,\ldots,0,1,\ldots,t_i$, and the overall enumeration is a concatenation of the sequences produced in phases $0,1,2,\ldots$. Hence, every negative integer appears at least once in this enumeration, and so does every positive integer, because $t_0 < t_1 < t_2 < \ldots$.
    
    Now, if $\mcA$ were to successfully exhaustively generate from $L_\infty$, for the above enumeration of $L_\infty$, there must be some finite time step $t_\infty$ such that for every $t \ge t_\infty$, the generator $\mcG_{t}$ output by $\mcA$ at $t$ satisfies that $S_{t} \cup Z_{< t} \cup Z_{\ge t} \supseteq L_\infty$. Consider then the smallest $i$ for which $t_i$ appears at a time step beyond $t_\infty$ (such an $i$ must exist, because the sequence $t_0,t_1,t_2,\ldots$ is increasing), and let this time step at which $t_i$ appears be $t'_i \geq t_\infty$. By the previous reasoning, we have that $S_{t'_{i}} \cup Z_{< t'_{i}} \cup Z_{\ge t'_{i}} \supseteq L_\infty$. However, by construction, the time step $t'_i$ is beyond $t^*_i$, and hence by the invariant that we have maintained (property (1) above), $|Z_{\ge t'_i} \setminus L_i| < \infty$. Hence, $S_{t'_i} \cup Z_{<_{t'_i}} \supseteq L_\infty \setminus (L_i \cup \{Z_{\ge t'_i} \setminus L_i\})$. But $S_{t'_i} \cup Z_{< t'_i}$ is a finite set, and  $L_\infty \setminus (L_i \cup \{Z_{\ge t'_i} \setminus L_i\})$ is an infinite set, so this is impossible.
\end{proof}

The proof above allows us to even rule out %
algorithms that satisfy much weaker versions of the exhaustive generation guarantee.
Consider, for example, the following randomized guarantee:

\begin{corollary}[Randomized Exhaustive Generation Lower Bound]
    For the collection $\mcC$ considered in the proof of \Cref{thm:exhaustive-generation-lb}, there cannot exist a randomized algorithm $\mcA$ which satisfies the following guarantee: for any $K \in \mcC$ and any enumeration of $K$, there exists $t^* < \infty$ such that for any $t \ge t^*$, it holds with probability $> 1/2$ that
    \begin{enumerate}
        \item $|Z_{\ge t} \setminus K|<\infty$.
        \item $S_t \cup Z_{< t} \cup Z_{\ge t} \supseteq K$.
    \end{enumerate}
\end{corollary}
\begin{proof}
    We need only modify the last paragraph in the proof of \Cref{thm:angluin-characterization} as follows: for the constructed enumeration of $L_\infty$, there must be some finite time step $t_\infty$ such that for every $t \ge t_\infty$, the generator $\mcG_{t}$ output by $\mcA$ at $t$ satisfies that with probability $> 1/2$, $S_{t} \cup Z_{< t} \cup Z_{\ge t} \supseteq L_\infty$. Consider then the smallest $i$ for which $t_i$ appears at a time step beyond $t_\infty$ (such an $i$ must exist, because the sequence $t_0,t_1,t_2,\ldots$ is increasing), and let this time step at which $t_i$ appears be $t'_i$. By the previous reasoning, we have that with probability $> 1/2$, $S_{t'_{i}} \cup Z_{< t'_{i}} \cup Z_{\ge t'_{i}} \supseteq L_\infty$. However, by construction, the time step $t'_i$ is beyond $t^*_i$, and hence by property (1) above, with probability $>1/2$, $|Z_{\ge t'_i} \setminus L_i|<\infty$.
    Since the two events $S_{t'_{i}} \cup Z_{< t'_{i}} \cup Z_{\ge t'_{i}} \supseteq L_\infty$ and $|Z_{\ge t'_i} \setminus L_i|<\infty$ individually occur with probability $>1/2$, they both together occur with probability $> 0$. Thus, with probability $> 0$,
    $S_{t'_i} \cup Z_{<_{t'_i}} \supseteq L_\infty \setminus (L_i \cup \{Z_{\ge t'_i} \setminus L_i\})$. But $S_{t'_i} \cup Z_{< t'_i}$ is a finite set, and $L_\infty \setminus (L_i \cup \{Z_{\ge t'_i} \setminus L_i\})$ is an infinite set, meaning that $S_{t'_{i}} \cup Z_{< t'_{i}} \cup Z_{\ge t'_{i}} \supseteq L_\infty$ should happen with probability $0$. This is a contradiction.
\end{proof}

\subsection{Separation between Identifiability and Exhaustive Generation}
\label{sec:identifiability-exhaustive-generation-separation}

One might observe that while it is not possible to exhaustively  generate from the collection $\mcC$ considered in the proof of \Cref{thm:exhaustive-generation-lb}, it is additionally also not possible to \textit{identify} from this collection. Furthermore, the proof we presented also has parallels to the way one might go about proving non-identifiability for this collection. A natural question to consider then is: are the notions of identifiability and exhaustive generation the same?

For clarity, we recall the definition of identifiability, wherein an algorithm is trying to figure out the \textit{index} of the target language being enumerated, and state Angluin's characterization \citep{angluin1979finding, angluin1980inductive} for collections of languages that are identifiable in the limit.  
\begin{definition}[Identification in the limit \cite{gold1967language}]
    An algorithm identifies in the limit from languages in a collection $\mcC=\{L_1,L_2,\ldots\}$, if for any language $K \in \mcC$ and any enumeration of $K$ presented to the algorithm, there exists $t^* < \infty$ such that for all $t \ge t^*$, the index $i_t$ output by the algorithm at time step $t$ satisfies $L_{i_t}=K$.
\end{definition}

\vspace{-.8cm}
\setlength{\jot}{0pt}
\begin{flalign}
    \label{condition:angluin-condition-with-enumeration}
    &\textbf{\underline{Angluin's Condition with Enumeration (Condition 1 in \cite{angluin1980inductive})}:}\nonumber \\
    &\text{A collection $\mcC=\{L_1,L_2,\ldots\}$ satisfies
    Angluin's Condition with Enumeration, if there exists a } \nonumber \\ 
    &\text{computable procedure, which for every language $L \in \mcC$, outputs an enumeration of a finite set $T$,} \nonumber \\
    &\text{such that $T \subseteq L$, and furthermore, every $L' \in \mcC$ that contains $T$ satisfies that $L'$ is not a proper }\nonumber \\
    &\text{subset of $L$.}&&
\end{flalign}

\begin{theorem}[Angluin's characterization (Theorem 1 in \cite{angluin1980inductive})]
    \label{thm:angluin-characterization}
    A collection of languages $\mcC=\{L_1,L_2,\ldots\}$ is identifiable in the limit if and only if it satisfies Angluin's Condition with Enumeration. 
\end{theorem}

Note that identifiability immediately implies exhaustive generation---once an algorithm has identified the target language, it can simply enumerate all the strings from it thereafter. Does exhaustive enumeration also imply identifiability? The following example shows that there are collections that are non-identifiable, but can be exhaustively generated.

\begin{example}
    \label{example:identifiability-exhaustive-generation-separation}
    Let $\Sigma^*=\Z$. Let $L_\infty = \Z$, and for any $i \in \Z$, let $L_i = \Z \setminus \{i\}$. Consider the countable collection $\mcC=L_\infty \cup \bigcup_{i \in \Z}L_i$. We first claim that it is not possible to identify from $\mcC$. To see this, observe that for every finite subset $T$ of $L_\infty$, any language $L_i$ for which $i \notin T$ contains $T$, but $L_i$ \textit{is} a proper subset of $L_\infty$. Hence, this collection cannot satisfy Angluin's Condition with Enumeration, and by \Cref{thm:angluin-characterization}, $\mcC$ is not identifiable in the limit.

    We now argue that it is possible to exhaustively generate from $\mcC$ in a straightforward manner. Consider the algorithm $\mcA$, which, oblivious of the input sequence, simply generates the sequence $0,-1,1,-2,2,-3,3,-4,4,\ldots$ in this order; the generator $\mcG_t:\N \to \Sigma^*$ that $\mcA$ outputs at time step $t$ may be appropriately deduced from this. We claim that this algorithm exhaustively generates from the collection. To see this, consider first that the target language is $L_\infty$. Then, we can set $t^*=1$, for which, we have that for all $t \ge t^*$, $Z_{\ge t} \subseteq L_\infty$, and also, $Z_{<t} \cup Z_{\ge t} = L_\infty$. Now, consider instead that the target language is $L_i$, for some $i \in \Z$. Then, observe crucially that once the algorithm generates $i$ (an error), it makes no further errors. Namely, let $t_i$ be the time step at which the algorithm generates $i$. We can set $t^*= t_i+1$, which guarantees that for all $t \ge t^*$, $Z_{\ge t} \subseteq L_i$ and also $Z_{<t} \cup Z_{\ge t} \supseteq L_i$.
\end{example}

\subsection{Connections to Generation with Breadth \citep{kalavasis2024limits}}
\label{sec:connections-to-breadth}

Our definition of exhaustive generation is related to the definition of generation with breadth given in \cite{kalavasis2024limits}. While both these definitions seek to formalize the notion of \textit{generating the entirety of the language}, there are important differences.

To elaborate on these, we first state the definition of generation with breadth in the limit from \cite{kalavasis2024limits}. While their model assumes that the generator output by the algorithm at time step $t$ is a \textit{distribution} over strings in $\Sigma^*$, and defines this notion in terms of the \textit{support} of the distribution, we can equivalently state the definition in our formulation, where the generator $\mcG_t$ is a mapping from $\N \to \Sigma^*$ that enumerates the support of the distribution. Conversely, we can obtain a generator in the sense of \cite{kalavasis2024limits} by considering any arbitrary distribution supported on the range $Z_{\ge t}$ of $\mcG_t$.

\begin{definition}[Generation with breadth \citep{kalavasis2024limits}]
    A generating algorithm $\mcA$ generates with breadth in the limit from languages in a collection $\mcC$, if for any $K \in \mcC$ and any enumeration of $K$, there exists $t^* < \infty$ such that for every $t \ge t^*$, it holds that $Z_{\ge t}=K$.\footnote{\cite{kalavasis2024limits} work with two equivalent definitions of generation with breadth; the condition we state is one of them (see Remark 2 in \cite{kalavasis2024limits}).}
\end{definition}

\begin{remark}
    \label{remark:kmv-other-notions}
    \cite{kalavasis2024limits} sometimes refer to the above as generation with \textit{exact} breadth, to disambiguate it from some other relaxations that they consider, like generation with \textit{approximate} breadth and \textit{unambiguous} generation. For approximate breadth (Definition 21 in \cite{kalavasis2024limits}), the generator is required to eventually only generate strings from the target language, but is allowed to miss out on \textit{finitely} many strings from it. In unambiguous generation (Definition 8 in \cite{kalavasis2024limits}), the generator is allowed to hallucinate (i.e., generate strings not belonging to the target language) infinitely often, but eventually, the set of strings it generates should have the smallest symmetric difference to the target language than any other language in the collection. For more details about how these, and other notions relate to one another, we refer the reader to the work by \cite{kalavasis2024characterizations}.
\end{remark}

\subsubsection{Consistency and Breadth Simultaneously Impossible}
\label{sec:consistency-breadth-impossible}

Observe that an algorithm that generates with breadth also satisfies our definition of exhaustive generation. This, together with \Cref{thm:exhaustive-generation-lb}, helps us answer an open question asked by \cite{kalavasis2024limits} in the negative:

\paragraph{Open Question 1 in \cite{kalavasis2024limits}:} \textit{Is there a class of generative algorithms for which the induced generators can be modeled as Turing machines and which achieve breadth and consistency for all countable collections of languages?}

Here, by ``consistency'', we simply mean the requirement of \Cref{def:generation-in-the-limit}. Note that if there was such a class of algorithms, which achieves generation with breadth for all countable collections, it would also imply a class of  algorithms, that achieves exhaustive generation for all countable collections. But this would contradict \Cref{thm:exhaustive-generation-lb}, which exhibits a countable collection that cannot be exhaustively generated by any algorithm. Thus, there cannot exist a class of generative algorithms that achieves breadth and consistency for all countable collections.

\subsubsection{Distinction between Generation with Breadth and Exhaustive Generation}
\label{sec:distinction-breadth-exhaustive-generation}

However, exhaustive generation \textit{does not} imply generation with breadth. In particular, notice that generation with breadth does not allow an algorithm to utilize the previous strings $Z_{\le t}$ generated by it, whereas our second condition for exhaustive generation (\Cref{def:exhaustive-generation}) allows an algorithm to include an arbitrary subset of $Z_{\le t}$ in order to cover the entirety of the target language. %
Our definition of exhaustive generation was formulated independently of the work of \cite{kalavasis2024limits}; hence the difference in the two definitions.
As we shall see, this crucial distinction is what separates exhaustive generation from generation with breadth.

Before we go on, we briefly explain our rationale for including a subset of previously generated elements $Z_{\le t}$ in covering the target language $K$.
Requiring the generator to produce the entirety of the target language (i.e., the condition $Z_{\ge t}=K$) does appear rather close (although not obviously equivalent) to being able to identify the target language.
The strong lower bounds on language identification motivated us to consider our definition of exhaustive generation, which allows for a notion of error in using previously generated elements to cover the target language $K$.
Allowing such errors is already present in
the original notion of generation in the limit from \cite{kleinberg2024language}.
The natural analog for exhaustive generation is to allow the algorithm to use a subset of the elements generated so far in covering $K$, leading to our definition.

We now elaborate on the relationship of the various notions: generation with breadth, exhaustive generation and language identification.
One of the main results (Theorem 3.5) in \cite{kalavasis2024limits} shows that, if there is a generating algorithm satisfying a certain ``$\mop$'' condition, which generates with breadth from a collection, then it can be used to construct an algorithm that can identify languages from the collection. The technical ``$\mop$'' condition, short for Membership Oracle Problem, is the following: for any generator that the algorithm may output at any time step (which, in the formulation of \cite{kalavasis2024limits}, is a \textit{distribution} over strings), it should be possible to decide whether any string $x$ belongs to the support of the generator.

But we can observe that the generators $\mcG_t$ output by the exhaustive generation algorithm from \Cref{example:identifiability-exhaustive-generation-separation}, viewed in the distributional sense of \cite{kalavasis2024limits}, satisfy the above $\mop$ property---to check if some string $x$ belongs to the support of $\mcG_t$, 
note that $Z_{\ge 2t+1} = \Z \setminus \{-t,\ldots,t\}$ and $Z_{\ge 2t} = \Z \setminus \{-t,\ldots, t-1\}$.
But %
even so, the collection in the example is not identifiable! We can also observe that if the target language is, say $L_1$, at no time step $t \ge 1$ in the generated sequence does it hold that $Z_{\ge t}=L_1$. Hence, this algorithm does not generate with breadth. In fact, as we show ahead, there cannot exist an algorithm that generates with breadth for this collection.

\begin{claim}
    \label{claim:breadth-exhaustive-gen-separation}
    The collection $\mcC$ considered in \Cref{example:identifiability-exhaustive-generation-separation} cannot be generated with breadth in the limit.
\end{claim}
\begin{proof}
    The proof steps are similar to the proof of \Cref{thm:exhaustive-generation-lb}. Assume for the sake of contradiction that there exists an algorithm $\mcA$ that generates the languages in $\mcC$ with breadth. This means that for any $K \in \mcC$ and any enumeration of $K$, there exists a $t^* < \infty$ such that for any $t \ge t^*$, it holds that $Z_{\ge t} = K$.
    
    For the sake of convenience, for any $x \in \Z$, define
    \begin{align*}
        \next(x) &= \begin{cases}
            -x & \text{if } x < 0 \\
            -(x+1) & \text{if } x > 0 \\
            -1 & \text{if } x =0.
        \end{cases}
    \end{align*}
    That is, $\next(x)$ is the example immediately following $x$ in the sequence $0,-1,1,-2,2,-3,3,\ldots$. Suppose now that an adversary enumerates $L_0$ as $-1,1,-2,2,-3,3,\ldots$. Then, there must exist some time step $t^*_0 < \infty$, such that for the generator $\mcG_{t^*_0}$ output by $\mcA$ at $t^*_0$, we have that $Z_{\ge t^*_0} = L_0$. Let $t_0=t^*_0$, and let the example in the enumeration above at time step $t_0$ be $n'_1$.

    Next, consider the enumeration of $L_{n_1}$, where $n_1=\next(n'_1)$, given as 
    $$-1,1,-2,2,-3,3,\ldots,n'_1, 0, -1, 1, -2, 2,\ldots,n'_1,\next(n_1),\ldots$$
    where beyond $n'_1$, observe that we inserted $0$, and skipped including $n_1$. The reason we inserted $0$ was because it was missing in the enumeration before $n'_1$ (on account of the previous enumeration of $L_0$). This is a valid enumeration of $L_{n_1}$, and hence there must exist a finite time step $t^*_{1} < \infty$, such that for any $t \ge t^*_{1}$, the generator $\mcG_{t}$ output by $\mcA$ at $t$ satisfies that $Z_{\ge t} = L_{n_1}$. In particular, choose a time step $t_1$ such that the example $n'_2$ at $t_1$ satisfies that $|n'_2| > |n'_1|$.

    Now, consider the enumeration of $L_{n_2}$, where $n_2=\next(n'_2)$, given as 
    $$-1,1,-2,2,-3,3,\ldots,n'_1, 0, -1, 1, -2, 2,\ldots,n'_2,n_1,0,-1,1,-2,2,-3,3,n'_2,\next(n_2),\ldots$$
    where beyond $n'_2$, observe that we again included $n_1$ (because it was missing in the enumeration till $n'_2$), and did not include $n_2$. This is a valid enumeration of $L_{n_2}$, and hence there must exist a finite time step $t^*_{2} < \infty$, such that for any $t \ge t^*_{2}$, the generator $\mcG_{t}$ output by $\mcA$ at $t$ satisfies that $Z_{\ge t} = L_{n_2}$. In particular, choose a time step $t_2$ such that the example $n'_3$ at $t_2$ satisfies that $|n'_3| > |n'_2|$.

    We can repeat the same argument indefinitely, and observe that this results in an enumeration of $L_\infty$. This is because the $|n'_1| < |n'_2| < |n'_3| < \ldots$, and the sequence between any $0,-1,1,\ldots,n'_{i+1}, n_i$ comprises of all the elements  in the enumeration $0,-1,1,-2,2,\ldots$ up until $n'_{i+1}$.

    Now, if $\mcA$ were to successfully generate with breadth from $L_\infty$, for the above enumeration of $L_\infty$, there must be some finite time step $t_\infty$ such that for every $t \ge t_\infty$, the generator $\mcG_{t}$ output by $\mcA$ at $t$ satisfies that $Z_{\ge t} = L_\infty$. Consider then the smallest $j$ such that $n'_{j+1}$ appears at a time step beyond $t_\infty$, and let $t_j$ be the time step at which $n'_{j+1}$ appears. By the invariant we have maintained above, it must be the case that $Z_{\ge t_j}=L_{n_j}$, which contradicts the requirement that $Z_{\ge t_j}=L_\infty$.
\end{proof}

We make one final summarizing remark: our lower bounds from \Cref{claim:breadth-exhaustive-gen-separation} and \Cref{thm:exhaustive-generation-lb} hold for \textit{all} generators, but only for the specific collections that we consider, whereas the lower bound from \cite[Theorem 3.5]{kalavasis2024limits} for generation with (exact) breadth holds for all non-identifiable collections, but only for a restricted class of generators (namely, those satisfying the MOP condition).

\subsubsection{Necessary Condition for Generation with Breadth}
\label{sec:generation-with-breadth-necessary-condition}

The collection in \Cref{example:identifiability-exhaustive-generation-separation} exhibits a countable collection for which generation with breadth is not possible, but exhaustive generation is possible. Recall that we had argued above that this collection is not identifiable in the limit (the same adversary strategy employed in the proof above also foils any given identification algorithm). As we show next, a necessary condition for identifiability, which we denote as ``Angluin's Condition with Existence'', is also necessary for generation with breadth. This further illustrates how the notions of identification in the limit and generation with breadth are closely tied together.

\vspace{-.6cm}
\setlength{\jot}{0pt}
\begin{flalign}
    \label{condition:angluin-condition-with-existence}
    &\textbf{\underline{Angluin's Condition with Existence (Condition 2 in \cite{angluin1980inductive})}:}\nonumber \\
    &\text{A collection $\mcC=\{L_1,L_2,\ldots\}$ satisfies
    Angluin's Condition with Existence, if for every language} \nonumber \\
    &\text{$L \in \mcC$, there exists a finite subset $T \subseteq L$, such that every $L' \in \mcC$ that contains $T$ satisfies that $L'$} \nonumber \\
    &\text{is not a proper subset of $L$.}&&
\end{flalign}

Corollary 1 in \cite{angluin1980inductive} shows that the above condition is necessary for identification in the limit. The following proposition shows that it is also necessary for generation with breadth.

\begin{proposition}
    \label{prop:generation-with-breadth-necessary-condition}
    If a collection $\mcC=\{L_1,L_2,\ldots\}$ can be generated with breadth, then it satisfies Angluin's Condition with Existence.
\end{proposition}
\begin{proof}
    Let us assume for the sake of contradiction that the collection $\mcC = \{L_1,L_2,\ldots\}$ does not satisfy \eqref{condition:angluin-condition-with-existence}, but $\mcA$ is an algorithm that generates from languages in $\mcC$ with breadth. Then, there exists a language $L \in \mcC$, such that 
    \begin{align}
    \text{$\forall$ finite subsets $T \subseteq L$, $\exists$ $L' \in \mcC$ that satisfies $T \subset L'$, and $L'$ is a proper subset of $L$.}
    \label{negation-condition-2}
    \end{align}
    Fix an enumeration of all the strings in $L$, and let this be 
    \begin{align}
        \label{eqn:fixed-ordering-L}
        L=\{x_1,x_2,\ldots\}.
    \end{align}
    Suppose that the adversary first presents $x_1$ as input. Then by \eqref{negation-condition-2}, for $T_1=\{x_1\}$, there exists $L_1 \in \mcC$ such that $T_1 \subset L_1$ and $L_1$ is a proper subset of $L$. Then, suppose that the adversary proceeds to enumerate strings in $L_1$ one by one in the order that they appear in the enumeration of $L$ above, skipping over $x_1$. Since $\mcA$ generates with breadth, and the strings presented so far constitute a valid enumeration of $L_1$, there must exist some time $t^*_1 < \infty$ such that for every $t \ge t^*_1$, it holds that $Z_{\ge t}=L_1$. Let $t_1 = t^*_1$. Suppose now that at time $t_1+1$, the adversary presents the string in $L \setminus L_1$ which has not been enumerated so far, and appears at the smallest index in the ordering in \eqref{eqn:fixed-ordering-L} above.

    Now, consider the finite set $S_{t_1+1}$ enumerated so far, and observe that $S_{t_1+1} \subset L$. So, let $T_2 = S_{t_1+1}$. Again by \eqref{negation-condition-2}, there exists $L_2 \in \mcC$ such that $T_2 \subset L_2$ and $L_2$ is a proper subset of $L$. Since all the strings presented so far are contained in $L_2$, from time step $t_1+2$ onward, the adversary continues to enumerate all the strings in $L_2$ one by one in the order that they appear in \eqref{eqn:fixed-ordering-L}, skipping over strings that have already been enumerated. Combined with the strings presented so far, this is a valid enumeration of $L_2$. Thus, there must exist some time $t^*_2 < \infty$ such that for every $t \ge t^*_2$, it holds that $Z_{\ge t}=L_2$. Let $t_2 \ge t^*_2$ be such that $t_2 > t_1$. Now suppose that at time $t_2+1$, the adversary presents the string in $L \setminus L_2$ which has not been enumerated so far, and appears at the smallest index in the ordering in \eqref{eqn:fixed-ordering-L} above.

    Now, consider the finite set $S_{t_2+1}$ enumerated so far, and observe that $S_{t_2+1} \subset L$. So, let $T_3 = S_{t_2+1}$. Again by \eqref{negation-condition-2}, there exists $L_3 \in \mcC$ such that $L_3 \supseteq T_3$, but also, $L_3$ is a proper subset of $L$. Thus, since all the strings presented so far are contained in $L_3$, from time step $t_2+2$ onward, the adversary continues to enumerate all the strings in $L_3$ one by one in the order that they appear in \eqref{eqn:fixed-ordering-L}, skipping over strings that have already been enumerated. Combined with the strings presented so far, this is a valid enumeration of $L_3$. Thus, there must exist some time $t^*_3 < \infty$ such that for every $t \ge t^*_3$, it holds that $Z_{\ge t}=L_3$. Let $t_3 \ge t^*_3$ be such that $t_3 > t_2$. Now suppose that at time $t_3+1$, the adversary inputs the string in $L \setminus L_3$ which has not been enumerated so far, and appears at the smallest index in the ordering in \eqref{eqn:fixed-ordering-L} above. \ldots

    We can repeat this argument indefinitely, and observe that the adversary will have produced an enumeration of $L$. This is because $t_1 < t_2 < t_3 \ldots$, and at every time step $t_i+1$, the adversary inputs the smallest indexed string in \eqref{eqn:fixed-ordering-L} that has not yet been enumerated.
    The condition that we maintain is that, when presented with the sequence up to time step $t_i$, 
    the generator $\mcG_{t_i}$ output by $\mcA$ at $t_i$ satisfies the generation with breadth guarantee for $L_i$; in particular, it maintains the invariant $Z_{\ge t_i} = L_{i}$.

    Now, if $\mcA$ were to successfully generate with breadth from $L$, for the above enumeration of $L$, there must be some finite time step $t^*$ such that for every $t \ge t^*$, the generator $\mcG_{t}$ output by $\mcA$ at $t$ satisfies that $Z_{\ge t} = L$. Consider then the smallest $j$ such that $t_j$ is a time step beyond $t^*$ in the above argument. By the invariant we have maintained, it must be the case that $Z_{\ge t_j}=L_{j}$. But note that $L_j$ is a proper subset of $L$, and this contradicts the generation with breadth requirement for $L$ that $Z_{\ge t_j}=L$.
\end{proof}

\begin{remark}
    \label{remark:generation-with-breadth-condition-gap}
    We note that the only technical difference between Angluin's Condition with Existence \eqref{condition:angluin-condition-with-existence} and Angluin's Condition with Enumeration \eqref{condition:angluin-condition-with-enumeration} is the \textit{efficient computability} (recursive enumerability) of the sets $T$ for every language $L$. Angluin's Condition with Enumeration, which requires the $T$'s to be computable, is also sufficient for generation with breadth, simply because it is sufficient for identification. However, we note that Angluin's Condition with Existence, which simply ensures existence of the $T$'s, is not sufficient for identification in the limit (see Theorem 2 in \cite{angluin1980inductive}). Therefore, if one were to show that Angluin's Condition with Enumeration (a stronger condition) is necessary for generation with breadth, then this would equate the two notions of identification in the limit and generation with breadth. On the other hand, showing that Angluin's Condition with Existence (a weaker condition) is sufficient for generation with breadth would separate the two notions.
\end{remark}

\section{Characterization of Exhaustive Generation}
\label{sec:exhaustive-generation-characterization}

Similar to Angluin's characterization (\Cref{thm:angluin-characterization}) for identification in the limit, in this section, we fully characterize exhaustive generation. First, in \Cref{sec:weak-angluin-condition-with-existence}, we introduce a weakening of Angluin's Condition with Existence, denoted as ``Weak Angluin's Condition with Existence'', and show that this condition characterizes the language conditions that can be exhaustively generated. Namely, \Cref{prop:exhaustive-characterization-necessary} shows that the condition is necessary for exhaustive generation to be possible at all. \Cref{prop:exhaustive-characterization-sufficient-function} shows that the condition is also sufficient for exhaustive generation, albeit with access to a somewhat strong oracle. Our next result (\Cref{prop:exhaustive-characterization-sufficient-membership}) in \Cref{sec:weak-angluin-condition-with-enumeration} then shows that a slight strengthening of the condition, which we denote as ``Weak Angluin's Condition with Enumeration'', is sufficient for exhaustive generation with the standard membership query oracle.

\subsection{Weakening of Angluin's Condition with Existence}
\label{sec:weak-angluin-condition-with-existence}

We start by showing that the collections of languages that can be exhaustively generated from are \textit{exactly} those collections that satisfy the following condition:

\vspace{-.6cm}
\setlength{\jot}{0pt}
\begin{flalign}
    \label{condition:weak-angluin-condition-with-existence}
    &\textbf{\underline{Weak Angluin's Condition with Existence}} \text{:}\nonumber \\
    &\text{A collection $\mcC=\{L_1,L_2,\ldots\}$ satisfies
    Weak Angluin's Condition with Existence, if for every } \nonumber \\
    &\text{language $L \in \mcC$, there exists a finite subset $T \subseteq L$, such that every $L' \in \mcC$ that contains $T$ and is } \nonumber \\
    &\text{a proper subset of $L$ satisfies $|L \setminus L'| < \infty$.}&&
\end{flalign}

The following two propositions, respectively show that the above condition is necessary and sufficient for exhaustive generation, thereby establishing \Cref{thm:exhaustive-generation-characterization}.

\begin{proposition}[Exhaustive Generation Necessary Condition]
    \label{prop:exhaustive-characterization-necessary}
    If a collection $\mcC=\{L_1,L_2,\ldots\}$ can be exhaustively generated, then it satisfies Weak Angluin's Condition with Existence.
\end{proposition}
\begin{proof}
    Let us assume for the sake of contradiction that the collection $\mcC = \{L_1,L_2,\ldots\}$ does not satisfy the condition, but $\mcA$ is an algorithm that exhaustively generates from languages in $\mcC$. Then, there exists a language $L \in \mcC$, such that: 
    \begin{align}
    \text{$\forall$ finite subsets $T \subset L$, $\exists$ $L' \in \mcC$, such that $T \subset L'$, $L' \subsetneq L$, but $|L \setminus L'|=\infty$.}
    \label{negation-of-exhaustive-generation-condition}
    \end{align}
    Fix an enumeration of all the strings in $L$, and let this be 
    \begin{align}
        \label{eqn:fixed-ordering-L-2}
        L=\{x_1,x_2,\ldots\}.
    \end{align}
    Suppose that the adversary first presents $x_1$ as input. Then by \eqref{negation-of-exhaustive-generation-condition}, for $T_1=\{x_1\}$, there exists $L_1 \in \mcC$ such that $T_1 \subset L_1$, $L_1$ is a proper subset of $L$, and $|L \setminus L_1|=\infty$. Then, suppose that the adversary proceeds to enumerate strings in $L_1$ one by one in the order that they appear in the enumeration of $L$ above, skipping over $x_1$. Since $\mcA$ generates exhaustively for the collection, and the strings presented so far constitute a valid enumeration of $L_1$, there must exist some time $t^*_1 < \infty$ such that for every $t \ge t^*_1$, it holds that $|Z_{\ge t} \setminus L_1|<\infty$. Let $t_1 = t^*_1$. Suppose now that at time $t_1+1$, the adversary presents the string in $L \setminus L_1$ that appears at the smallest index in the ordering in \eqref{eqn:fixed-ordering-L-2} within the set of strings that have not been enumerated so far. 

    Now, consider the finite set $S_{t_1+1}$ enumerated so far, and observe that $S_{t_1+1} \subseteq L$. Let $T_2 = S_{t_1+1}$. Again by \eqref{negation-of-exhaustive-generation-condition}, there exists $L_2 \in \mcC$ such that $T_2 \subset L_2$, $L_2$ is a proper subset of $L$, and $|L \setminus L_2|=\infty$. Since all the strings presented so far are contained in $L_2$, from time step $t_1+2$ onward, the adversary continues to enumerate all the strings in $L_2$ one by one in the order that they appear in \eqref{eqn:fixed-ordering-L-2}, skipping over strings that have already been enumerated. Combined with the strings presented so far, this is a valid enumeration of $L_2$. Thus, there must exist some time $t^*_2 < \infty$ such that for every $t \ge t^*_2$, it holds that $|Z_{\ge t} \setminus L_2| < \infty$. Let $t_2 \ge t^*_2$ be such that $t_2 > t_1$. Now suppose that at time $t_2+1$, the adversary presents the string in $L \setminus L_2$ that appears at the smallest index in the ordering in \eqref{eqn:fixed-ordering-L-2} within the set of strings which have not been enumerated so far.

    Now, consider the finite set $S_{t_2+1}$ enumerated so far, and observe that $S_{t_2+1} \subseteq L$. Let $T_3 = S_{t_2+1}$. Again by \eqref{negation-of-exhaustive-generation-condition}, there exists $L_3 \in \mcC$ such that $T_3 \subset L_3$, $L_3$ is a proper subset of $L$, and $|L \setminus L_3|=\infty$. Since all the strings presented so far are contained in $L_3$, from time step $t_2+2$ onward, the adversary continues to enumerate all the strings in $L_3$ one by one in the order that they appear in \eqref{eqn:fixed-ordering-L-2}, skipping over strings that have already been enumerated. Combined with the strings presented so far, this is a valid enumeration of $L_3$. Thus, there must exist some time $t^*_3 < \infty$ such that for every $t \ge t^*_3$, it holds that $|Z_{\ge t} \setminus L_3|<\infty$. Let $t_3 \ge t^*_3$ be such that $t_3 > t_2$. Now suppose that at time $t_3+1$, the adversary presents the string in $L \setminus L_3$ which appears at the smallest index in the ordering in \eqref{eqn:fixed-ordering-L-2} within the set of strings that have not been enumerated so far. \ldots

    We can repeat this argument indefinitely, and observe that the adversary will have produced an enumeration of $L$. This is because $t_1 < t_2 < t_3 \ldots$, and at every time step $t_i+1$, the adversary presents the string with smallest index in the ordering \eqref{eqn:fixed-ordering-L-2} that has not yet been enumerated.
    The condition that we maintain is that, when presented with the sequence up to time step $t_i$, 
    the generator $\mcG_{t_i}$ output by $\mcA$ at $t_i$ satisfies the exhaustive generation guarantee for $L_i$; in particular, it maintains the invariant $|Z_{\ge t_i} \setminus L_{i}|<\infty$.
    
    Now, if $\mcA$ were to exhaustively generate from $L$, for the above enumeration of $L$, there must be some finite time step $t^*$ such that for every $t \ge t^*$, the generator $\mcG_{t}$ output by $\mcA$ at $t$ satisfies that $S_t \cup Z_{<t} \cup Z_{\ge t} \supseteq L$. Consider then the smallest $j$ such that $t_j$ is a time step beyond $t^*$ in the above argument. By the exhaustive generation requirement for $L$, we must have that $S_{t_j} \cup Z_{< t_j} \cup Z_{\ge t_j} \supseteq L$. However, by the invariant we have maintained above, it is also the case that $|Z_{\ge t_j} \setminus L_{j}|<\infty$. Hence,
    \begin{align}
    S_{t_j} \cup Z_{< t_j} \supseteq L \setminus Z_{\ge t_j} \supseteq L \setminus (L_j \cup Z_{\ge t_j}) = L \setminus (L_j \cup (Z_{\ge t_j} \setminus L_{j})) = (L \setminus L_j) \setminus (Z_{\ge t_j} \setminus L_{j}).
    \label{finite-infinite-set-equality}
    \end{align}
    But $S_{t_j} \cup Z_{t_j}$ and $Z_{\ge t_j} \setminus L_{j}$ are finite sets, while $L \setminus L_j$ is an infinite set, so \eqref{finite-infinite-set-equality} is impossible.    
\end{proof}

\begin{remark}[Relaxed Exhaustive Generation]
    \label{remark:relaxed-exhaustive-generation}
    We remark that essentially the same proof above also establishes that Weak Angluin's Condition with Existence is a necessary condition for a slightly more relaxed definition of exhaustive generation. Namely, consider replacing the two conditions in \Cref{def:exhaustive-generation}, by the single condition $|Z_{\ge t} ~\Delta~ K| < \infty$, where $\Delta$ denotes the symmetric difference (i.e., $A ~\Delta~ B=(A\setminus B) \cup (B \setminus A)$). We can verify that the same contradiction in the proof above also goes through for this definition. However, we also observe that while \Cref{def:exhaustive-generation} implies this relaxed definition, the relaxed definition does not imply \Cref{def:exhaustive-generation} (namely, the relaxed definition does not necessarily satisfy the second condition in \Cref{def:exhaustive-generation}).
\end{remark}

\begin{proposition}[Exhaustive Generation Sufficient Condition]
    \label{prop:exhaustive-characterization-sufficient-function}
    If a collection $\mcC=\{L_1,L_2,\ldots\}$ satisfies Weak Angluin's Condition with Existence, then it can be exhaustively generated, 
    with access to an oracle which determines, for any $i,j$, whether $L_i \setminus L_j$ is finite.
\end{proposition}
\begin{proof}
    We will show that the algorithm of \cite{kleinberg2024language}, together with a slight modification, exhaustively generates in the limit. First, we recall their algorithm. Suppose that the target language being enumerated is $L_{z}$ for $z \in \N$.

    At time step $t$, the algorithm considers the languages in the subcollection $\mcC_t = \{L_i: 1 \le i \le t, L_i \supseteq S_t\}$, i.e., the languages among $L_1,\ldots,L_t$ that are consistent with the input $S_t$ enumerated so far. Let $\mcC_t = \{L_{i_1},\ldots,L_{i_{n_t}} \}$, where $i_1 < i_2 < \ldots < i_{n_t}$. A language $L_{i_j} \in \mcC_t$ is termed \textit{critical} if $L_{i_j} \subseteq L_{i_k}$ for every $k < j$. Consider the largest $j \le n_t$ such that the language $L_{i_j} \in \mcC_t$ is critical, and denote this $j$ by $t^\star$---the algorithm generates a string from $L_{i_{t^\star}} \setminus S_t$. 
    
    We quickly review their proof of correctness for this algorithm. We can verify that there exists a large enough time step $t^+ \ge z$ such that for all $t \ge t^+$, language $L_{z}$ is critical at time $t$ (this is because $L_{z}$ is always consistent with the input, and hence always belongs to $\mcC_t$ for $t \ge z$, and also, for every language $L_i$ such that $i < z$ which satisfies that $L_{z} \not\subseteq L_i$, there is a string in $L_{z} \setminus L_i$ which eventually gets enumerated, and hence makes $L_i$ inconsistent). Furthermore, we can also verify that for every $t \ge t^+$, the last critical language $L_{i_{t^\star}}$ that the algorithm chooses to generate from satisfies (by definition of criticality) that $L_{i_{t^\star}} \subseteq L_{z}$. This establishes correctness. %

    Our exhaustive generation algorithm closely follows this algorithm of \cite{kleinberg2024language}. At time step $t$, let $L_{i_{t^\star}}$ be the language considered as above. Then, as we argued, for every $t \ge t^+$, we have that $L_{i_{t^\star}} \subseteq L_{z}$. Now, consider the finite set $T_{z} \subseteq L_{z}$ which is guaranteed to exist by \eqref{condition:weak-angluin-condition-with-existence}. Observe crucially that there also exists a large enough $t'$ such that for every $t \ge t'$, $T_{z} \subseteq S_t$. This is because $T_{z}$ is a finite subset of $L_{z}$, and every string in $L_{z}$ is guaranteed to show up in the input enumeration at some finite time. So, consider $t''=\max(t', t^+)$. 

    We now claim that for every $t \ge t''$, the language $L_{i_{t^\star}}$, in addition to satisfying $L_{i_{t^\star}} \subseteq L_z$ (by the property of criticality), also additionally satisfies that $|L_z \setminus L_{i_{t^\star}}| < \infty$. To see this, suppose that $L_{i_{t^\star}} \subset L_z$ (if $L_{i_{t^\star}} = L_z$, then the claim is vacuously true). Because we have chosen a time step larger than $t'$, the argument from the previous paragraph gives us that $T_z \subseteq S_t$. Furthermore, by definition of the algorithm, $S_t \subseteq L_{i_{t^\star}}$. Thus, $L_{i_{t^\star}}$ is a language satisfying $T_z \subseteq L_{i_{t^\star}}$, and also that $L_{i_{t^\star}} \subset L_z$. Therefore, \eqref{condition:weak-angluin-condition-with-existence} implies that $|L_z \setminus L_{i_{t^\star}}| < \infty$.

    So, consider populating a language $Z_{\ge t}$ as follows. We initialize $Z_{\ge t}=L_{i_{t^\star}}$. In the collection of consistent languages $\mcC_t = \{L_{i_1},\ldots, L_{i_{t^\star}},\ldots, L_{n_t}\}$ considered at time step $t$, consider every $j \le t^\star$ such that 
    \begin{align}
    \text{$L_{i_j}$ is critical (and also a superset of $L_{i_{t^\star}}$), and also satisfies $|L_{i_j} \setminus L_{i_{t^\star}}| < \infty$,}
    \label{checked-condition}
    \end{align}
    and update $Z_{\ge t} = Z_{\ge t} \cup (L_{i_j} \setminus L_{i_{t^\star}})$ for every such $j$. 
    
    Note that because $L_z \in \mcC_t$, and $z \le i_{t^\star}$, there is a $j$ satisfying $i_j=z$ that will be considered. Furthermore, by the argument above, $|L_z \setminus L_{i_{t^\star}}| < \infty$, and hence this $j$ will pass the condition \eqref{checked-condition}. Because there are only finitely many $j \le t$, and for every $j$ satisfying the condition, we only add finitely many strings to $Z_{\ge t}$, at the end of the procedure we ensure that 1) $Z_{\ge t} \supseteq L_z$, and also 2) $|Z_{\ge t} \setminus L_z| < \infty$.

    The exhaustive generation algorithm outputs the generator $\mcG_t$, which simply enumerates the strings in $Z_{\ge t}$ that is constructed by the above procedure, if it is asked to go into generate-only mode at this time. By design, we have ensured that for all $t \ge t''$, both conditions required for exhaustive generation  in \Cref{def:exhaustive-generation} are satisfied.
\end{proof}

\begin{remark}
    We note that the algorithm in the above proof of \Cref{prop:exhaustive-characterization-sufficient-function} also needs access to a {\em subset oracle} which, given indices $i,j$ determines whether $L_i \subseteq L_j$.
    \citet{kleinberg2024language} show how to implement their algorithm using only membership queries, without needing such a subset oracle. The same ideas can be applied in our setting as well.
\end{remark}

\subsection{Weakening of Angluin's Condition with Enumeration}
\label{sec:weak-angluin-condition-with-enumeration}

While the condition in \eqref{condition:weak-angluin-condition-with-existence} fully characterizes exhaustive generation, our algorithm above for collections satisfying this condition admittedly requires a more powerful oracle (for any $i,j$, it must be able to determine if $L_i \setminus L_j$ is finite) than the standard membership query oracle considered by \cite{kleinberg2024language}. As our final result in this section, we show that a slight strengthening of \eqref{condition:weak-angluin-condition-with-existence} is sufficient for exhaustive generation with the standard membership query oracle. The strengthened condition requires the efficient computatability of the tell-tale sets in \eqref{condition:weak-angluin-condition-with-existence}, and is yet another instantiation of the subtle difference between existence and enumerability of these sets.

\vspace{-.6cm}
\setlength{\jot}{0pt}
\begin{flalign}
    \label{condition:weak-angluin-condition-with-enumeration}
    &\textbf{\underline{Weak Angluin's Condition with Enumeration}} \text{:}\nonumber \\
    &\text{A collection $\mcC=\{L_1,L_2,\ldots\}$ satisfies
    Weak Angluin's Condition with Enumeration, if there exists} \nonumber \\
    &\text{a computable procedure, which for every language $L \in \mcC$, ouptuts an enumeration of a finite set} \nonumber \\
    &\text{$T$, such that $T \subseteq L$, and furthermore, every $L' \in \mcC$ that contains $T$ and is a proper subset of $L$} \nonumber \\
    &\text{satisfies $|L \setminus L'| < \infty$.}&&
\end{flalign}

\begin{proposition}%
    \label{prop:exhaustive-characterization-sufficient-membership}
    If a collection $\mcC=\{L_1,L_2,\ldots\}$ satisfies Weak Angluin's Condition with Enumeration, then it can be exhaustively generated with only membership oracle access to the language collection.
\end{proposition}
\begin{proof}
Suppose we have a collection $\mcC=\{L_1,L_2,\ldots\}$ that satisfies \eqref{condition:weak-angluin-condition-with-enumeration}.
We describe an algorithm for exhaustive generation with only membership query access to the language collection, inspired by Angluin's algorithm for language identification (in the proof of Theorem 1 in \citep{angluin1980inductive}). 
Let $T_i^{(n)}$ denote the set of strings produced in the first $n$ steps of the enumeration of $T_i$.

At time step $n$, the algorithm considers the languages in the subcollection $\mcC_n = \{L_i: 1 \le i \le n, L_i \supseteq S_n \supseteq T_i^{(n)}\}$, i.e., the languages $L_i$ among $L_1,\ldots,L_n$ that are consistent with the input $S_n$ enumerated so far, with the additional condition that all strings in $T_i^{(n)}$ have also appeared in the input. %
It is easy to check that $\mcC_n$ can be determined with only membership oracle access to the language collection.
If $\mcC_n$ is empty, the algorithm sets $Z_{\geq n}$ arbitrarily.
Otherwise, let $g$ be the smallest index of a language in $\mcC_n$. The algorithm sets $Z_{\geq n}$ to be $L_g$.

We now establish correctness for this algorithm.
Suppose that the target language being enumerated is $L_{z}$ for $z \in \N$. For each $i \in \{1,\ldots,z\}$, we define $n_i$ as follows:
If $L_z \setminus L_i \neq \emptyset$, then $n_i$ is the first time step when a string from $L_z \setminus L_i$ appears in the input.
Otherwise, if $L_i \supseteq L_z$, $n_i$ is the smallest value of $n$ such that $T_i^{(n)} = T_i$.
(Note that, in fact, $T_i^{(n)} = T_i$ for all $n \geq n_i$).

Consider any time step $n \geq \max \{n_i, 1\leq i \leq z\}$.
We will show that the algorithm satisfies the correctness guarantee for exhaustive generation.
We claim that the subcollection $\mcC_n \cap \{L_1, \ldots L_z\}$ consists of precisely those languages $L_i, 1 \leq i \leq z$ such that $L_i \supseteq L_z$ and $L_i \setminus L_z$ is finite. Note that this also implies that $\mcC_n$ is non-empty since $L_z$ satisfies this condition.

Suppose that $L_z \setminus L_i \neq \emptyset$.
Since $n \geq n_i$, $S_n$ contains a string in $L_z \setminus L_i$.
Thus $L_i$ is not a consistent language at time step $n$ and is not included in $\mcC_n$.
On the other hand, suppose $L_i \in \mcC_n$.
By the previous argument, $L_i \supseteq L_z$.
Since $n \geq n_i$, $T_i \subseteq T_i^{(n)} \subseteq S_n \subseteq L_z$.
(The first set inclusion follows from the definition of $n_i$, the second from the definition of $\mcC_n$, and the third from the fact that $L_z$ is the target language.) 
Applying \eqref{condition:weak-angluin-condition-with-enumeration} with $L' = L_z$, we conclude that $|L_i \setminus L_z| < \infty$.

Now consider the smallest index $g$ of a language in $\mcC_n$.
Recall that $g$ exists since $\mcC_n$ is non-empty.
By the claim we just established, $L_g \supseteq L_z$ and $L_g \setminus L_z$ is finite.
Since the algorithm sets $Z_{\geq n} = L_g$, the algorithm satisfies the requirement for exhaustive generation.
\end{proof}

\section{Uniform Generation with Feedback}
\label{sec:genfeed}

In this section, we consider the setting of uniform generation from a language collection \textit{with feedback}, where the generating algorithm is additionally allowed, at each time step, to query if any string $w$ of choice belongs to the target language $K$ being enumerated. 
We note again that this model is different from the membership query model considered in \cite{kleinberg2024language}; there, an algorithm can only query if a string $w$ belongs to any language $L_i$ in the collection $\mcC=\{L_1,L_2,\ldots\}$. 
In the feedback setting, the algorithm is more powerful since it can directly query membership in the target language $K$.

As before, we restrict our attention to countable language collections---we know by \Cref{thm:non-uniform-ub} that non-uniform generation is always possible for such collections, without requiring any feedback. So, we want to further understand when \textit{uniform} generation is possible, with or without feedback. It is easy to construct examples of language collections that cannot be uniformly generated without feedback, but can be uniformly generated with feedback. The following example is adopted from Lemma 15 in \cite{raman2024generation}.

\begin{example}
    \label{eg:gf-but-no-gnf}
    Consider a partition $\{S_d\}_{d \in \N}$ of $\N$, where $|S_1| < |S_2| < \ldots$. Let $E$ be the set of all negative even integers, and $O$ be the set of all negative odd integers. For $d \in \N$, let $L^E_d = E \cup S_d$, and let $L^O_d = O \cup S_d$, and consider the language collection $\mcC$ comprising of all languages $L^E_d$ and $L^O_d$ for $d \in \N$. This collection has infinite closure dimension (which can be seen by noticing that the intersection of languages containing any $S_d$---namely $L^E_d$ and $L^O_d$---is exactly $S_d$), and hence cannot be uniformly generated from without feedback. However, observe that with just one query, a generator can find out whether the target language belongs to the ``even'' or ``odd'' category. Once it knows this, it can generate indefinitely from either the set of even or odd negative integers.
\end{example}

We want to now identify a property of a given collection $\mcC$
which characterizes whether it is possible to uniformly generate with feedback.
Our goal is to formulate such a property in the form of a combinatorial dimension like the closure dimension (Definition \ref{def:closure-dimension}) from the work of \cite{raman2024generation} for uniform generation.

Towards this, we first propose an alternate combinatorial dimension 
which also characterizes uniform generation (without feedback) 
and show that this is equivalent to the closure dimension.
We later generalize this combinatorial dimension and show that it characterizes uniform generation with feedback.

\subsection{The $\gnf$ Dimension for Generation with no Feedback}

First, we abstractly formalize some notions from the language generation setup, which will be particularly helpful in defining the dimension.

\paragraph{Transcript:} A transcript is a record of interaction between an adversary (who is enumerating a language $K \in \mcC$) and a generator (who is trying to generate from the language). A transcript is an infinite sequence $x_1,z_1,x_2,z_2,x_3,z_3,\ldots$, where each $x_t \in K$; here, $x_t$ is the input given to the generator at time step $t$, and $z_t$ is the string generated by the generator at time step $t$. Note that the actions of the adversary and the generator are interleaved.

\paragraph{Adversary Strategy:} An adversary strategy $A$ is a mapping from prefixes of a transcript ending in an action by the generator (in this case, the last string generated by the generator) to the next action by the adversary (in this case, the next string from $K$ to append to the enumeration). For any language $L \in \mcC$, an adversary strategy $A$ is consistent with $L$ if the sequence $x_1,x_2,\ldots$ is an enumeration of \textit{all} the strings in $L$, i.e., for every $x \in L$, there is some finite index $i$ such that $x_i=x$. Here again, we use the shorthand $S_t$ to denote the set of distinct strings in the sequence $x_1,\dots,x_t$.

\paragraph{Generator Strategy:} A generator strategy $G$ is a mapping from prefixes of a transcript ending in an action by the adversary to an action by the generator.

For an adversary strategy $A$ and generator strategy $G$, the transcript of interaction between $A$ and $G$, denoted $T(A,G)$, is the infinite sequence $x_1,z_1,x_2,z_2,\ldots$, where
\begin{align*}
    &x_1 = A(\emptyset) \\
    &z_1 = G(x_1) \\
    &x_2 = A(x_1, z_1) \\
    &z_2 = G(x_1,z_1,x_2) \\
    &\vdots
\end{align*}

\paragraph{Consistent Languages:} For a collection $\mcC$ of languages, generator strategy $G$, and adversary strategy $A$ consistent with some language $K \in \mcC$, consider the transcript $T=T(A,G)$. For $r \in \N$, we say that a language $L \in \mcC$ is consistent with $T$ upto round $r$ if $x_t \in L$ for every $t \le r$. Let $\mcC_r(T)$ denote the subset of $\mcC$ comprising of languages consistent with $T$ upto round $r$. Note that $K \in \mcC_r(T)$ for every $r \in \N$, and hence $\mcC_r(T)$ is never empty.

\paragraph{Effective Intersection:} We define the effective intersection at round $r$, denoted $E_r(T)$, as
\begin{align}
    \label{eqn:effective-intersection}
    E_r(T) &= \left\{\bigcap_{L \in \mcC_r(T)}L\right\} \setminus S_r.
\end{align}

We are now ready to define a complexity measure which we term the Generation-no-Feedback ($\gnf$) dimension.
\begin{definition}[$\gnf$ dimension]
    \label{def:gnf-dimension}
    The $\gnf$ dimension of a collection $\mcC$ is the supremum over $d \in \N$, for $d$ satisfying the following property: for every generator strategy $G$, there exists a language $K \in \mcC$ and an adversary strategy $A$ consistent with $K$ such that, in the transcript $T=T(A,G)$, there exists a finite $r \ge d$ where $|S_r| \ge d$ and $E_r(T) = \emptyset$. 
\end{definition}

We first show that the $\gnf$ dimension characterizes uniform generation, by equating it to the closure dimension (Definition \ref{def:closure-dimension}) from the work of \cite{raman2024generation}.

\begin{proposition}
    For any collection $\mcC$, the $\gnf$ dimension of $\mcC$ is equal to its closure dimension.
\end{proposition}
\begin{proof}
    Consider any $1 \le d < \infty$, and suppose that the $\gnf$ dimension of $\mcC$ is at least $d$. This means that there exists $d' \ge d$, such that for every generator strategy $G$, there exists a language $K \in \mcC$ and an adversary strategy $A$ consistent with $K$ such that in the transcript $T=T(A,G)$, there exists a finite $r \ge d'$ such that $|S_r| \ge d'$ and $E_r(T)=\emptyset$. %
    Then, fix an arbitrary generator strategy, and consider the $K \in \mcC$ and adversary strategy $A$ consistent with $K$ that satisfy this property. In particular, consider the time step $r$ in the transcript $T(A,G)$ at which $|S_r| \ge d'$ and $E_r(T)=\emptyset$. Recall that $\mcC_r(T)$ is not empty (it contains $K$), and furthermore, every language in $\mcC_r(T)$ contains $S_r$. Hence, since $E_r(T)=\emptyset$, the definition of $E_r(T)$ (see \eqref{eqn:effective-intersection}) implies that    
    \begin{align*}
        \bigcap_{L \in \mcC_r(T)}L = S_r.
    \end{align*}
    But note that the languages in $\mcC_r(T)$ are exactly those languages that are consistent with $S_r$ (i.e., they contain $S_r$), and furthermore, $S_r$ is a finite set. Thus, we have that the languages in $\mcC$ containing $S_r$ have finite intersection. Furthermore, $S_r$ is a set of size at least $d' \ge d$. This implies that the closure dimension of $\mcC$ is at least $d' \ge d$.

    Now, suppose that the closure dimension of $\mcC$ is at least $d$. This means that we can find a set $S$ of size $d' \ge d$, such that the intersection of languages in $\mcC$ that contain $S$ is finite. In particular, let
    \begin{align}
        \label{eqn:supset-condition}
        \bigcap_{L \in \mcC, L \supseteq S}L = \{x_1,\ldots,x_{d''}\},
    \end{align}
    where $d'\le d'' < \infty$. Now, fix any generator strategy $G$. Choose any $L \in \mcC$ that satisfies $L \supseteq S$ to be the target language $K$---we know that $K$ contains $\{x_1,\ldots,x_{d''}\}$ from \eqref{eqn:supset-condition}. Then, consider the adversary strategy $A$ which, in the first $d''$ rounds, enumerates $x_1,\ldots,x_{d''}$, and then continues to enumerate the rest of the strings in $K$, irrespective of the the generator's actions. Then, we have that $S_{d''}=\{x_1,\ldots,x_{d''}\}$, and $|S_{d''}|=d'' \ge d' \ge d$. Furthermore, we claim that $\cap_{L \in \mcC_{d''}(T)}L=S_{d''}$, which would imply that $E_{d''}(T)=\emptyset$. To see this, observe that $\mcC_{d''}(T)=\{L \in \mcC: L \supseteq S_{d''}\}$. Let $\mcC_1 = \{L \in \mcC: L \supseteq S\}$. Observe that if $L \supseteq S$, $L \supseteq \{x_1,\ldots,x_{d''}\}=S_{d''}$ by \eqref{eqn:supset-condition}. Thus, $\mcC_1 \subseteq \mcC_{d''}(T)$. Furthermore, if $L \supseteq S_{d''}$, then $L \supseteq S$, since $S_{d''} \supseteq S$. Thus, we also have that $\mcC_{d''}(T) \subseteq \mcC_1$, implying that $\mcC_{d''}(T) = \mcC_1$. This gives that 
    \begin{align*}
        \bigcap_{L \in \mcC_{d''}(T)}L = \bigcap_{L \in \mcC_1}L=\{x_1,\ldots,x_{d''}\} = S_{d''},
    \end{align*}
    where in the second equality, we used \eqref{eqn:supset-condition} again.
    Thus, we have obtained an $r = d'' \ge d$ such that $|S_r|\ge d''\ge d$ and $E_r(T) = \emptyset$. Since this holds regardless of the generator strategy, we have that the $\gnf$ dimension of $\mcC$ is at least $d'' \geq d$.

    The above argument shows that, for any finite $d \ge 1$, the closure dimension of $\mcC$ is at least $d$ if and only if the $\gnf$ dimension of $\mcC$ is at least $d$. Thus, either both the closure and $\gnf$ dimension of $\mcC$ are finite and equal, or they are both unbounded.
\end{proof}

Therefore, the $\gnf$ dimension is equivalent to the closure dimension and hence, characterizes uniform generation. However, this alternate formulation of the closure dimension in terms of abstract adversary and generator strategies allows us to generalize to the case where the generator is additionally allowed to query membership of a string in the target language at each time step.

\subsection{The $\gf$ Dimension for Generation with Feedback}

We systematically extend the notions defined in the previous section to account for queries issued by the generator.

\paragraph{Transcript:} A transcript is now an infinite sequence $(x_t, y_t, a_t, z_t)_{t \in \N}$, here, at time step $t$, $x_t$ is the input given by the adversary to the generator, $y_t \in \Sigma^*$ is the query issued by the generator, $a_t \in \{\yes, \no\}$ is the response given by the adversary to the membership query, and $z_t$ is the string generated by the generator. Note that the actions of the adversary and the generator are still interleaved.

\paragraph{Adversary Strategy:} An adversary strategy $A$ is still a mapping from prefixes of a transcript ending in an action by the generator (either the last string generated, or the membership query issued by the generator) to the next action by the adversary (either the next string to append to the enumeration, or a $\yes$/$\no$ response). For any language $L \in \mcC$, an adversary strategy $A$ is consistent with $L$ if (1) the sequence $x_1,x_2,\ldots$ is an enumeration of \textit{all} the strings in $L$, i.e., for every $x \in L$, there is some finite index $t$ such that $x_t=x$, and (2) for all $t \in \N$, the response $a_t=\yes$ if $y_t \in L$, and $\no$ otherwise. The shorthand $S_t$ still denotes the set of distinct strings in the input sequence $x_1,\dots,x_t$.

\paragraph{Generator Strategy:} Similarly, a generator strategy $G$ is still a mapping from prefixes of a transcript ending in an action by the adversary to an action by the generator.

For an adversary strategy $A$ and generator strategy $G$, the transcript $T(A,G)$ of interaction between $A$ and $G$ is now the infinite sequence $(x_t, y_t, a_t, z_t)_{t \in \N}$, where
\begin{align*}
    &x_1 = A(\emptyset) \\
    &y_1 = G(x_1) \\
    &a_1 = A(x_1, y_1) \\
    &z_1 = G(x_1, y_1, a_1) \\
    &x_2 = A(x_1, y_1, a_1, z_1) \\
    &y_2 = G(x_1, y_1, a_1, z_1, x_2) \\
    &a_2 = A(x_1, y_1, a_1, z_1, x_2, y_2) \\
    &z_2 = G(x_1, y_1, a_1, z_1, x_2, y_2, a_2) \\
    &\vdots
\end{align*}

We can now state the definition for uniform generation with feedback.

\begin{definition}[Uniform Generation with Feedback]
    \label{def:uniform-generation-with-feedback}
     A collection $\mcC$ can be uniformly generated from with feedback if there exists a generator strategy $G$ and a constant $t^*=t^*(\mcC)$, such that for every language $K \in \mcC$ and for every adversary strategy $A$ consistent with $K$, in the transcript $T(A,G)$, $z_t \in K \setminus S_t$ for every $t$ satisfying $|S_t| \ge t^*$.
\end{definition}

\paragraph{Consistent Languages:} For a collection $\mcC$ of languages, generator strategy $G$, and adversary strategy $A$ consistent with some language $K \in \mcC$, consider the transcript $T=T(A,G)$. For $r \in \N$, we say that a language $L \in \mcC$ is consistent with $T$ upto round $r$ if (1) $x_t \in L$ for every $t \le r$, and (2) $a_t=\yes$ if $y_t \in L$, and $\no$ otherwise for every $t \le r$. Let $\mcC_r(T)$ denote the subset of $\mcC$ comprising of languages consistent with $T$ upto round $r$. Note that $K \in \mcC_r(T)$ for every $r \in \N$, and hence $\mcC_r(T)$ is never empty.

\paragraph{Effective Intersection:} The definition of the effective intersection $E_r(T)$ at round $r$ remains the same as given in \eqref{eqn:effective-intersection}.

With these changes to the notions of adversary and generator strategies, the $\gf$ dimension is defined similarly as in \Cref{def:gnf-dimension}:
\begin{definition}[$\gf$ dimension]
    \label{def:gf-dimension}
    The $\gf$ dimension of a collection $\mcC$ is the supremum over $d \in \N$, for $d$ satisfying the following property: for every generator strategy $G$, there exists a language $K \in \mcC$ and an adversary strategy $A$ consistent with $K$, such that in the transcript $T=T(A,G)$, there exists a finite $r \ge d$ where $|S_r| \ge d$ and $E_r(T) = \emptyset$. 
\end{definition}

\begin{remark}
    One reason behind stating our definition of $\gf$ dimension directly in terms of adversary/generator strategies, and not in terms of a property about intersection of languages containing a set of strings as in \cite{raman2024generation}, is because of the ability of the generator to make membership queries on previously unseen strings. Essentially, there is additional complexity in controlling the intersection of consistent languages, which may drastically change based on queries that the generator may ask.
\end{remark}

We showed that the $\gnf$ dimension characterizes uniform generation (without feedback), by relating it to the closure dimension. Here, we directly argue that the $\gf$ dimension characterizes uniform generation with feedback. This follows from the following two lemmas:

\begin{lemma}[$\gf$ Dimension Upper Bound]
    If the $\gf$ dimension of a collection $\mcC$ is finite, it can be uniformly generated with feedback.
\end{lemma}
\begin{proof}
    Suppose the $\gf$ dimension of $\mcC$ is $d < \infty$. This means that for every $d' > d$, the property in \Cref{def:gf-dimension} is \textit{not} satisfied for $d'$. In particular, substituting $d'=d+1$, we get that: there exists a generator strategy $G$, such that for every language $K \in \mcC$ and adversary strategy $A$ consistent with $K$, in the transcript $T=T(A,G)$, for every $r \ge d+1$, either $|S_r| < d+1$ or $E_r(T)\not=\emptyset$. But note that any adversary strategy consistent with $K$ must eventually satisfy $|S_{r}| \ge d+1$ for all large enough $r \ge d+1$---this follows from the requirement that the adversary must enumerate all the strings in $K$. So, fix the first such $r' \ge d+1$ where $|S_{r'}| = d+1$. Then, for all $r \ge r'$, this generator strategy satisfies that $E_r(T) \neq \emptyset$, which means that there exists a string in $K \setminus S_r$ that the generator can generate. Thus, the collection $\mcC$ can be uniformly generated with feedback by this generator with $t^*(\mcC)=d+1$.
\end{proof}

\begin{lemma}[$\gf$ Dimension Lower Bound]
    If the $\gf$ dimension of a collection $\mcC$ is infinite, it cannot be uniformly generated with feedback.
\end{lemma}
\begin{proof}
    Suppose a generator strategy $G'$ claims to uniformly generate from languages in $\mcC$ (with feedback) as soon as it sees $t^*=t^*(\mcC)$ distinct strings. We will show that there is a language $L \in \mcC$, and an adversary strategy $A'$ consistent with $L$, such that in the transcript $T(A', G')$, the string generated by $G'$ at some time step $r$ where $|S_r| \ge t^*$ does not belong to $L \setminus S_r$. %
    This would violate the uniform generation guarantee of $G'$.

    Towards this, note that because the $\gf$ dimension of $\mcC$ is infinite, we can find some $d \ge t^*$, such that the following property holds: for every generator strategy $G$, there exists a language $K \in \mcC$ and an adversary strategy $A$ consistent with $K$, such that in the transcript $T=T(A,G)$, there exists a finite $r \ge d$ where $|S_r| \ge d$ and $E_r(T) = \emptyset$. So, let us choose the language $K \in \mcC$ and adversary strategy $A$ consistent with $K$ for the particular generator $G'$ from above. Then, we know that for some $r \ge d \ge t^*$, it is the case that $|S_r| \ge d \ge t^*$, and $E_r(T(A, G'))=\emptyset$. In particular, consider the collection $\mcC_r(T(A, G'))$---this collection is non-empty because $K$ belongs to it. Let $z_r$ be the string generated by $G'$ at time step $r$. If $z_r \in S_r$, $G'$ is not a valid generator. So, suppose that $z_r \notin S_r$. Because $E_r(T(A, G'))=\emptyset$, this means that
    \begin{align*}
        \bigcap_{L \in \mcC_r(T(A, G'))}L = S_r.
    \end{align*}
    Thus, we also have that $z_r \notin \bigcap_{L \in \mcC_r(T(A, G'))}L$. In particular, this means that there exists some $L \in \mcC_r(T(A, G'))$ such that $z_r \notin L$. But now, consider the adversary strategy $A'$, which makes the same actions as $A$ up until time step $r$, but from time step $r+1$ onward, continues to arbitrarily complete the enumeration of $L$, responding to any queries according to membership in $L$. Note that $A'$ is consistent with $L$, since $L \in \mcC_r(T(A, G'))$, which means that all the strings enumerated as well as the answers to any membership queries given by $A$ up until $r$ are consistent with every language in $\mcC_r(T(A, G'))$. Furthermore, the transcript $T(A', G')$ is identical to $T(A, G')$ in the first $r$ rounds. However, %
    $z_r \notin L \setminus S_r$. 
    Thus, we have obtained a language $L$ and an adversarial strategy $A'$ consistent with $L$, such that at a time $r \ge t^*$ where $|S_r| \ge t^*$, in the transcript $T(A', G')$, the string $z_r$ generated by $G'$ does not belong to $L \setminus S_r$. This contradicts the uniform generation guarantee of $G'$.
\end{proof}

Recall that \Cref{eg:gf-but-no-gnf} demonstrated a collection where uniform generation without feedback was not possible, but uniform generation with feedback was. We conclude this section with an example of a language collection that cannot be uniformly generated with or without feedback, but can still be non-uniformly generated.

\begin{example}
    \label{eg:L-exclude-i}
    Let $\Sigma^* = \Z$. For any $i \in \Z$, let $L_i = \Z \setminus \{i\}$, and for any $i, j \in \Z$ such that $i < j$, let $L_{ij}=\Z \setminus \{i,j\}$. Consider the (countable) collection $\mcC=\bigcup_{i \in \Z}L_i \cup \bigcup_{i < j \in \Z}L_{ij}$.
We will argue that $\mcC$ cannot be uniformly generated, with or without feedback, but can be non-uniformly generated.

$\mcC$ cannot be uniformly generated without feedback because $\mcC$ has infinite closure dimension---this can be seen, for example, by observing that for any $d \ge 1$, the intersection of languages in $\mcC$ containing the set $\{1,2,\ldots,d\}$ is exactly $\{1,2,\ldots,d\}$. 

$\mcC$ cannot be uniformly generated even \textit{with} feedback, because $\mcC$ has infinite $\gf$ dimension. To see this, fix any $d \ge 1$, and generator strategy $G$. Consider the adversary strategy $A$ that operates as follows: let $x_t$ be the example input by the adversary, and $y_t$ be the example queried by the generator at time step $t$. Furthermore, let $S_t$ and $Q_t$ denote the sets of distinct elements in $x_1,\dots,x_t$ and $y_1,\dots,y_t$ respectively. For $t=1,\dots,d$, the input $x_t$ is decided as follows: if $y_{t-1} \in S_{t-1}$, then $x_t$ is a distinct number not contained in $S_{t-1}$. Else if $y_{t-1} \notin S_{t-1}$, then $x_{t}=y_t$. Next, for $t=1,\dots,d-1$, the answers $a_t$ to all queries are $\yes$. However, the answer $a_d$ to the last query is decided based on $y_d$: if $y_d \in S_d$, then $a_d=\yes$, else if $y_d \notin S_d$, $a_d = \no$. This specifies the adversary strategy up until $d$ rounds. Note that $|S_d|=d$ by construction. We must now choose a language $K \in \mcC$ consistent with this strategy, and then specify how the adversary strategy continues beyond round $d$. 
    
First, consider the case that the adversary answered the last query with a $\yes$. Observe in this case that $Q_d \subseteq S_d$, and recall that we answered all queries in $Q_d$ with a $\yes$. We choose $K$ to be any language in $\mcC$ that contains $S_d$, and the adversary $A$ continues to enumerate $K$ (and answers any queries according to $K$) beyond round $d$. Then, in the transcript $T=T(A,G)$ generated, $\mcC_d(T)$ comprises of all the languages in $\mcC$ that do not exclude any element in $S_d$. Furthermore, the intersection of all these languages is exactly $S_d$, because any element outside of $S_d$ is excluded by at least one language in $\mcC_d(T)$. This means that $E_d(T)=\emptyset$.

Next, consider the case that the adversary answered the last query with a $\no$. This means that the last query $y_d$ that the generator issued does not belong to $S_d$. Note however that $Q_{d-1} \subseteq S_d$, and $A$ answered all the queries in $Q_{d-1}$ with a $\yes$. Hence, we must choose a language in $\mcC$ that contains all of $S_d$, but excludes $y_d$. Choose the target language $K$ to be any $L_{ij} \in \mcC$ that contains $S_d$, and for which $i=y_d$. The adversary continues to enumerate $K$ (and answers any queries according to $K$) beyond round $d$. Then, in the transcript $T=T(A,G)$ generated, $\mcC_d(T)$ comprises of all the languages in $\mcC$ that contain $S_d$, but exclude $y_d$. In particular, $\mcC_d(T)$ comprises of all languages $L_{ij}$ that contain $S_d$, for which one of $i$ or $j$ is equal to $y_d$, and the other is any number in $\Z \setminus S_d$. This means that the intersection of all languages in $\mcC_d(T)$ is again exactly $S_d$, which means that $E_d(T)=\emptyset$.

Since the above argument holds for any $d \ge 1$, the $\gf$ dimension of $\mcC$ is infinite. 

Finally, we argue that $\mcC$ can be non-uniformly generated (without any feedback). This follows from \Cref{thm:non-uniform-ub}, and by the fact that $\mcC$ is countable. Even more directly, consider the generator that simply generates the sequence $0,-1,1,-2,2,-3,3,-4,4,\ldots$, while always skipping any elements that have shown up in the input. Observe that there is a finite time step $t(\mcC, K)$ (which is $O(|i|)$ if the true language $K$ is $L_i$, or $O(\max(|i|,|j|))$ if it is $L_{ij}$) beyond which the generator is sure to have ``skipped past'' the excluded elements in $K$, independent of its order of enumeration. Thus, such a generator non-uniformly generates from the collection.

\end{example}

\section*{Acknowledgements}
The authors are supported by Moses Charikar's Simons Investigator Award. Chirag is additionally also supported by Gregory Valiant's Simons Investigator Award.

\bibliographystyle{plainnat}
\bibliography{references}

\begin{thebibliography}{17}
\providecommand{\natexlab}[1]{#1}
\providecommand{\url}[1]{\texttt{#1}}
\expandafter\ifx\csname urlstyle\endcsname\relax
  \providecommand{\doi}[1]{doi: #1}\else
  \providecommand{\doi}{doi: \begingroup \urlstyle{rm}\Url}\fi

\bibitem[Angluin(1979)]{angluin1979finding}
Dana Angluin.
\newblock Finding patterns common to a set of strings.
\newblock In \emph{Proceedings of the eleventh annual ACM Symposium on Theory of Computing}, pages 130--141, 1979.

\bibitem[Angluin(1980)]{angluin1980inductive}
Dana Angluin.
\newblock Inductive inference of formal languages from positive data.
\newblock \emph{Information and control}, 45\penalty0 (2):\penalty0 117--135, 1980.

\bibitem[Arjovsky and Bottou(2017)]{arjovsky2017towards}
Martin Arjovsky and L{\'e}on Bottou.
\newblock Towards principled methods for training generative adversarial networks.
\newblock \emph{arXiv preprint arXiv:1701.04862}, 2017.

\bibitem[Arjovsky et~al.(2017)Arjovsky, Chintala, and Bottou]{arjovsky2017wasserstein}
Martin Arjovsky, Soumith Chintala, and L{\'e}on Bottou.
\newblock Wasserstein generative adversarial networks.
\newblock In \emph{International conference on machine learning}, pages 214--223. PMLR, 2017.

\bibitem[Banerjee et~al.(2024)Banerjee, Agarwal, and Singla]{banerjee2024llms}
Sourav Banerjee, Ayushi Agarwal, and Saloni Singla.
\newblock Llms will always hallucinate, and we need to live with this.
\newblock \emph{arXiv preprint arXiv:2409.05746}, 2024.

\bibitem[Bousquet et~al.(2021)Bousquet, Hanneke, Moran, Van~Handel, and Yehudayoff]{bousquet2021theory}
Olivier Bousquet, Steve Hanneke, Shay Moran, Ramon Van~Handel, and Amir Yehudayoff.
\newblock A theory of universal learning.
\newblock In \emph{Proceedings of the 53rd Annual ACM SIGACT Symposium on Theory of Computing}, pages 532--541, 2021.

\bibitem[Charikar and Pabbaraju(2024)]{charikar2024exploring}
Moses Charikar and Chirag Pabbaraju.
\newblock Exploring facets of language generation in the limit.
\newblock \emph{arXiv preprint arXiv:2411.15364v1}, 2024.

\bibitem[Christiano et~al.(2017)Christiano, Leike, Brown, Martic, Legg, and Amodei]{christiano2017deep}
Paul~F Christiano, Jan Leike, Tom Brown, Miljan Martic, Shane Legg, and Dario Amodei.
\newblock Deep reinforcement learning from human preferences.
\newblock \emph{Advances in neural information processing systems}, 30, 2017.

\bibitem[Gold(1967)]{gold1967language}
E~Mark Gold.
\newblock Language identification in the limit.
\newblock \emph{Information and control}, 10\penalty0 (5):\penalty0 447--474, 1967.

\bibitem[Kalai and Vempala(2024)]{kalai2024calibrated}
Adam~Tauman Kalai and Santosh~S Vempala.
\newblock Calibrated language models must hallucinate.
\newblock In \emph{Proceedings of the 56th Annual ACM Symposium on Theory of Computing}, pages 160--171, 2024.

\bibitem[Kalavasis et~al.(2024{\natexlab{a}})Kalavasis, Mehrotra, and Velegkas]{kalavasis2024characterizations}
Alkis Kalavasis, Anay Mehrotra, and Grigoris Velegkas.
\newblock Characterizations of language generation with breadth.
\newblock \emph{arXiv preprint}, 2024{\natexlab{a}}.

\bibitem[Kalavasis et~al.(2024{\natexlab{b}})Kalavasis, Mehrotra, and Velegkas]{kalavasis2024limits}
Alkis Kalavasis, Anay Mehrotra, and Grigoris Velegkas.
\newblock On the limits of language generation: Trade-offs between hallucination and mode collapse.
\newblock \emph{arXiv preprint arXiv:2411.09642}, 2024{\natexlab{b}}.

\bibitem[Kleinberg and Mullainathan(2024)]{kleinberg2024language}
Jon Kleinberg and Sendhil Mullainathan.
\newblock Language generation in the limit.
\newblock \emph{arXiv preprint arXiv:2404.06757}, 2024.

\bibitem[Li et~al.(2024)Li, Raman, and Tewari]{li2024generation}
Jiaxun Li, Vinod Raman, and Ambuj Tewari.
\newblock Generation through the lens of learning theory.
\newblock \emph{arXiv preprint arXiv:2410.13714v4}, 2024.

\bibitem[Raman and Tewari(2024)]{raman2024generation}
Vinod Raman and Ambuj Tewari.
\newblock Generation through the lens of learning theory.
\newblock \emph{arXiv preprint arXiv:2410.13714v3}, 2024.

\bibitem[Wu et~al.(2024)Wu, Grama, and Szpankowski]{wu2024no}
Changlong Wu, Ananth Grama, and Wojciech Szpankowski.
\newblock No free lunch: Fundamental limits of learning non-hallucinating generative models.
\newblock \emph{arXiv preprint arXiv:2410.19217}, 2024.

\bibitem[Xu et~al.(2024)Xu, Jain, and Kankanhalli]{xu2024hallucination}
Ziwei Xu, Sanjay Jain, and Mohan Kankanhalli.
\newblock Hallucination is inevitable: An innate limitation of large language models.
\newblock \emph{arXiv preprint arXiv:2401.11817}, 2024.

\end{thebibliography}

\appendix
    \section{Comparison to \cite{kalavasis2024characterizations}}
    \label{sec:kmv-characterizations-comparison}
    
    To aid the reader, we map some of our results to the results in \cite{kalavasis2024characterizations}. 
    
    \paragraph{Exhaustive Generation:} Our result showing that Weak Angluin's Condition with Existence (\Cref{prop:exhaustive-characterization-necessary}) is necessary for exhaustive generation is comparable to the similar result in Lemma 2.11 in \cite{kalavasis2024characterizations}. Our result showing the sufficiency of Weak Angluin's Condition with Existence (\Cref{prop:exhaustive-characterization-sufficient-function}) for exhaustive generation is comparable to Lemma 2.9 in \cite{kalavasis2024characterizations}. Our result showing the sufficiency of Weak Angluin's Condition with Enumeration (\Cref{prop:exhaustive-characterization-sufficient-function}) for exhaustive generation with only membership queries is comparable to Lemma 2.10 in \cite{kalavasis2024characterizations}.
    
    \paragraph{Generation with (Exact) Breadth:} Our result showing that Angluin's Condition with Existence (\Cref{prop:generation-with-breadth-necessary-condition}) is necessary for generation with (exact) breadth is comparable to Lemma 2.1 in \cite{kalavasis2024characterizations}.

\end{document}